\documentclass{article}

\usepackage{amssymb}
\usepackage{latexsym}
\usepackage{amsmath}
\usepackage{amsthm}
\usepackage{graphicx}

\usepackage{color}

\usepackage{a4wide}
\usepackage{cite}

\newcommand{\imag}{{\Im\mbox{\rm m}}}
\newcommand{\dom}{{\mbox{\rm dom}}}
\newcommand{\ran}{\mathrm{ran}}

\newcommand{\supp}{{\mbox{\rm supp}}}

\newcommand{\clospa}{{\mathrm{clospan}}}

\newcommand{\clo}{{\mathrm{clo}}}
\newcommand{\sign}{{\mathrm{sign}}}

\newcommand{\tr}{{\mathrm{tr}}}

\newcommand{\slim}{\,\mbox{\em s-}\hspace{-2pt} \lim}
\newcommand{\wlim}{\,\mbox{\em w-}\hspace{-2pt} \lim}
\newcommand{\olim}{\,\mbox{\em o-}\hspace{-2pt} \lim}

\newtheorem{theorem}{Theorem}[section]
\newtheorem{lemma}[theorem]{Lemma}
\newtheorem{corollary}[theorem]{Corollary}
\newtheorem{proposition}[theorem]{Proposition}

\newtheorem{assumption}[theorem]{Assumption}
\newtheorem{definition}[theorem]{Definition}
\newtheorem{example}[theorem]{Example}

\newcommand{\ba}{\begin{array}}
\newcommand{\ea}{\end{array}}
\newcommand{\bea}{\begin{eqnarray}}
\newcommand{\eea}{\end{eqnarray}}
\newcommand{\bead}{\begin{eqnarray*}}
\newcommand{\eead}{\end{eqnarray*}}
\newcommand{\be}{\begin{equation}}
\newcommand{\ee}{\end{equation}}
\newcommand{\bed}{\begin{displaymath}}
\newcommand{\eed}{\end{displaymath}}
\newcommand{\bal}{\begin{align}}
\newcommand{\eal}{\end{align}}

\newcommand{\bl}{\begin{lemma}}
\newcommand{\el}{\end{lemma}}
\newcommand{\bt}{\begin{theorem}}
\newcommand{\et}{\end{theorem}}
\newcommand{\bd}{\begin{definition}}
\newcommand{\ed}{\end{definition}}
\newcommand{\bc}{\begin{corollary}}
\newcommand{\ec}{\end{corollary}}
\newcommand{\bass}{\begin{assumption}}
\newcommand{\eass}{\end{assumption}}
\newcommand{\bexam}{\begin{example}}
\newcommand{\eexam}{\end{example}}

\newcommand{\Label}{\label}
\newcommand{\la}{\Label}

\newcommand{\de}{\textnormal{d}}

\def\wt#1{{{\widetilde #1} }}
\def\wh#1{{{\,\widehat #1\,} }}

   \def\sB{{\mathfrak B}}   
\def\sD{{\mathfrak D}}      
   \def\sH{{\mathfrak H}}   
   \def\sK{{\mathfrak K}}   \def\sL{{\mathfrak L}}

\def\sS{{\mathfrak S}}

   \def\se{{\mathfrak e}}   
   \def\sh{{\mathfrak h}}   
   \def\sk{{\mathfrak k}}

      \def\dC{{\mathbb C}}
\def\dD{{\mathbb D}}   \def\dE{{\mathbb E}}

   \def\dN{{\mathbb N}}   
      \def\dR{{\mathbb R}}
   \def\dT{{\mathbb T}}

   \def\cB{{\mathcal B}}   
      \def\cF{{\mathcal F}}
   \def\cH{{\mathcal H}}   
\def\cJ{{\mathcal J}}      
\def\cM{{\mathcal M}}      
      
\def\cS{{\mathcal S}}      \def\cU{{\mathcal U}}
      
   \def\cZ{{\mathcal Z}}

\def\gD{{\Delta}}			\def\ga{{\alpha}}
\def\gG{{\Gamma}}			\def\gb{{\beta}}
\def\gl{{\lambda}}
\def\gL{{\Lambda}}			\def\gga{{\gamma}}
\def\gO{{\Omega}}			\def\gd{{\delta}}
\def\gs{{\sigma}}			
\def\gY{{\Upsilon}}                     \def\gk{{\kappa}}
\def\gth{\theta}                        \def\gT{{\Theta}}
\def\gS{{\Sigma}}
\def\gr{{\varrho}}

\def\ess-sup{{\text{ess-sup\,}}}
\def\supp{{\text{\rm supp\,}}}

\numberwithin{equation}{section}

\title{Cayley transform applied to non-interacting quantum transport}

\author{Horia D. Cornean\thanks{E-mail: cornean@math.aau.dk}\\
Department of Mathematical Sciences, Aalborg University\\
Fredrik Bajers Vej 7, 9220 Aalborg, Denmark
\and
Hagen Neidhardt\thanks{E-mail: hagen.neidhardt@wias-berlin.de},
Lukas Wilhelm\thanks{E-mail: lukas.wilhelm@wias-berlin.de}\\
WIAS Berlin, Mohrenstr. 39, 10117 Berlin,
Germany
\and 
Valentin A. Zagrebnov\thanks{E-mail: Valentin.Zagrebnov@latp.univ-mrs.fr}\\
D\'{e}partement de Math\'{e}matiques\\
Universit\'{e} d'Aix-Marseille\\
and\\
Laboratoire d'Analyse, Topologie, Probaliti\'{e}s (UMR 7353)\\
Centre de Math\'{e}matique et Informatique-AMU\\
Technop\^{o}le Ch\^{a}teau-Gombert\\
39, rue F. Joliot Curie, 13453 Marseille Cedex 13\\
France}

\date{\today}

\begin{document}

\maketitle

\vspace{-5mm}
\begin{abstract}
\noindent
We extend the Landauer-B\"uttiker formalism in order to accommodate both unitary and 
self-adjoint operators which are not bounded from below. We also prove
that the pure point and singular continuous subspaces of the decoupled
Hamiltonian do not contribute to the steady current. One of the physical
applications is a stationary charge current formula for a system with
four pseudo-relativistic semi-infinite leads and with an inner sample
which is described by a Schr\"odinger operator defined on a bounded
interval with dissipative boundary conditions. Another application is
a current formula for electrons  described by a one dimensional Dirac
operator; here the system consists of  two semi-infinite leads coupled through a point interaction at zero.
\end{abstract}



\newpage 

\section{Introduction}

Considering a problem in quantum statistical mechanics and solid
state physics Lifshits \cite{Lif1952} found
that there is a unique real-valued function $\xi(\cdot) \in L^1(\dR,d\gl)$ such
that the formula
\be\la{Lif}
\tr(\Phi(H_0 + V) - \Phi(H_0)) = \int_\dR \xi(\gl) \Phi'(\gl)d\gl
\ee
is valid for a suitable class of functions $\Phi(\cdot)$
guaranteeing that $\Phi(H_0 + V) - \Phi(H_0)$ is a trace class
operator. Here $H_0$ is a self-adjoint operator and $V$ is a finite dimensional
self-adjoint operator.  Formula \eqref{Lif} and function $\xi(\cdot)$ 
are known in the literature as trace formula and 
spectral shift function, respectively.

Inspired by the work of Lifshits the trace formula was 
carefully investigated and generalized by Krein, cf. 
\cite{Krein1953}. In a first step Krein has shown that
Lifshits' result remains true if $V$ is a self-adjoint trace class operator. Later on he
generalized the result to pairs of self-adjoint operators $\cS =
\{H,H_0\}$ such that their resolvent difference is a trace class
operator, cf. \cite{Krein1962}. In the following we call those pairs
trace class scattering systems. For trace class scattering systems there exists a real-valued function
$\xi(\cdot) \in L^1(\dR,\tfrac{d\gl}{1 + \gl^2})$ called also the
spectral shift function such that 
\be\la{Krein}
\tr\left(\Phi(H) - \Phi(H_0)\right) = \int_\dR \xi(\gl)\Phi'(\gl)d\gl
\ee
is valid for a suitable class of  functions $\Phi(\cdot)$. 
In particular, the formula
\bed
\tr\left((H-z)^{-1} - (H_0-z)^{-1}\right) = -\int_\dR
\frac{\xi(\gl)}{(\gl-z)^2}d\gl, \quad z \in \dC \setminus \dR,
\eed
holds.  In contrast to the spectral shift function
from above $\xi(\cdot)$ is now not unique and is only determined up to a real constant.
To verify \eqref{Krein} Krein firstly proved a
trace formula \eqref{Lif} for a pair $\cU = \{U,U_0\}$ of unitary operators for which 
$U - U_0$ is a trace class operator, cf. \cite{Krein1962}. Regarding $U$ and $U_0$
as the Cayley transforms of $H$ and $H_0$, respectively, Krein was
able to establish \eqref{Krein}. 
 
If $\cS = \{H,H_0\}$ is a trace class scattering system, then the wave operators
\be\la{wave}
W_\pm(H,H_0) = \slim_{t\to\pm\infty}e^{itH}e^{-itH_0}P^{ac}(H_0)
\ee
exist and are complete where $P^{ac}(H_0)$
is the projection onto the absolutely continuous subspace of $H_0$,
see \cite{Birman1962}. 
 Let $\Pi(H^{ac}_0)$ be a spectral representation of the absolutely
continuous part $H^{ac}_0$ of $H_0$, cf. Appendix \ref{C}. Further, let
$\{S(\gl)\}_{\gl\in\dR}$ be 
the scattering matrix of the trace class scattering system $\cS$ with
respect to $\Pi(H^{ac}_0)$.  It turns out 
that there is a suitable chosen spectral shift function $\xi(\cdot)$ 
such that the so-called Birman-Krein formula
\bed
\det(S(\gl)) = e^{-2\pi i \xi(\gl)}.
\eed
holds for a.e. $\gl \in \dR$.

The quantity $T(\gl) := \tfrac{1}{2\pi i}(I_{\sh(\gl)}-S(\gl))$, $\gl\in \dR$, is usually called 
the transition matrix, see \eqref{2.62a}, where $I_{\sh(\gl)}$ denotes
the fiber identity operator of the spectral representation $\Pi(H^{ac}_0)$. 
In \cite{Radul1990} Radulescu has shown that the transition matrix $\{T(\gl)\}_{\gl \in \dR}$,
the unperturbed operator $H_0$ and 
the perturbation $V$ are related in a certain way. Indeed, if $H_0$ is bounded and 
$V$ is trace class, then the formula
\bed
\tr(H^n_0W_+(H,H_0)V) = \int_\dR \gl^n\,\tr(T(\gl))d\gl, \quad n = 0,1,2,\ldots,
\eed
is valid.

It turns out that the so-called Landauer-B\"uttiker formula
is a further interesting example in this circle of relations linking scattering matrix,
unperturbed operator and perturbation. From the physical point of view
the Landauer-B\"uttiker formula gives the steady
state charge current flowing trough a non-relativistic quantum
device where the carriers are not self-interacting. It goes back to
Landauer and B\"uttiker, cf. 
\cite{Landauer1957} and \cite{Buettiker1985}, and was initially
derived by them using phenomenological arguments. 

The physical setting is as follows: 
there is a small sample (the inner system) and at least two leads 
(for simplicity we only discuss the two lead case). At negative times,  
the leads are not coupled to the inner system. Each subsystem is in a state of thermal equilibrium.  
In particular, one assumes that in the leads the electrons are 
distributed according to the Fermi-Dirac distribution function. More
precisely, if $\mu_j$ are the chemical potentials of the left and
right leads, $j \in \{l,r\}$, then the energy distribution of lead $j$ is $f_j(\gl) = f_{FD}(\gl-\mu_j)$ where:
\be\la{0.1}
f_{FD}(\gl) = \frac{1}{1 + e^{\gb \gl}}, \quad \gl \in \dR, \quad \gb > 0.
\ee
At time zero the leads are suddenly attached to the inner system
and a current can flow from one lead to the other through the inner system. Landauer found by heuristic arguments (later refined by 
B\"uttiker) that the stationary current $J$ of non-relativistic particles flowing through the system should be given by 
\be\la{0.2}
J = \frac{\se}{2\pi} \int_\dR d\gl \,|\gs(\gl)|^2(\lambda) \big(f_{FD}(\lambda - \mu_l) - f_{FD}(\lambda-\mu_r)\big) 
\ee
where $\gs(\gl)$ is the so-called transmission coefficient between the leads, a 
cross-section arising from an appropriate scattering system, and
$\se > 0$ is the magnitude of the elementary charge.
The current is directed from left to right if $J > 0$ and from right
to left if $J < 0$. If $\mu_l > \mu_r$, then a straightforward
computation shows that $J > 0$ which shows that the charge current is directed from
the higher chemical potential to the lower one. 

Several works have already been published in which this approach has 
been made rigorous, cf. \cite{Cornean2005,Pillet2007,Cornean2006,Nenciu2007,
Nenciu2008,Cornean2009,Cornean2010b}. One assumes that at 
negative times the system is described by (a decoupled) Hamiltonian
$H_0$, while for positive times by (a coupled Hamiltonian) $H$. Until
now it was always assumed that both Hamiltonians are bounded from below
and that the difference between their resolvents raised to some
integer power is trace class.

Since our paper only deals with operator theoretical aspects of 
quantum transport of quasi-free particles, some of the terminology
used in quantum statistical mechanics will be strictly adapted to 
our limited needs. For us, a {\it density operator} is just any 
non-negative bounded operator. A density operator $\rho$ is an 
{\it equilibrium state of $H_0$} if it is a positive function of
$H_0$. A density operator $\rho$ is called a {\it steady state of $H_0$} if 
$\rho$ commutes with $H_0$.  Note that with our definition,
equilibrium states are steady states. If $H_0$ is  a decoupled direct sum of 
several operators $\bigoplus h_j$, then a direct sum of individual
equilibrium states $\bigoplus F_j(h_j)$ would provide us with a special class of steady states of $H_0$.  

A {\it charge} is any bounded self-adjoint operator $Q$ commuting with $H_0$. 
Following \cite{Pillet2007}, the {\it steady current} $J^\cS_{\rho,Q}$ 
related to a charge $Q$ and a given initial steady state $\rho$ of  $H_0$ is proved to be given by
\be\la{0.3}
J^\cS_{\rho,Q} := -i\tr(W_-(H,H_0)\rho W_-(H,H_0)^*[H,Q])
\ee
provided the commutator $[H,Q]$ is well defined and $H$ has no singular continuous spectrum. 
Following \cite{Pillet2007} the current is directed from the leads to the sample.
If the commutator is
not well defined, a regularization procedure was proposed in
\cite{Pillet2007}.  It consists in replacing the operators $H$ and $H_0$ by
bounded self-adjoint operators 
\be\la{0.4}
H(\eta) := H(I + \eta H)^{-N}
\quad \mbox{and} \quad
H_0(\eta) := H_0(I + \eta H_0)^{-N}, \quad \eta > 0,
\ee
for some large enough $N$, where for simplicity it is assumed that both operators are
non-negative. Of course, $\cS(\eta) = \{H(\eta),H_0(\eta)\}$ is  also
a trace class scattering system for which the current
$J^{\cS(\eta)}_{\rho,Q}$ is well defined. Finally, one sets
\be\la{0.5}
J^{\cS}_{\rho,Q} := \lim_{\eta\to+0}J^{\cS(\eta)}_{\rho,Q}.
\ee
We note that the absolutely subspace $\sH^{ac}(H_0)$ reduces the
initial steady state and the charge operator. Let
\be\la{0.6}
\rho_{ac} := \rho \upharpoonright\sH^{ac}(H_0) 
\quad \mbox{and} \quad 
Q_{ac} := Q \upharpoonright\sH^{ac}(H_0)
\ee
The restrictions $\rho_{ac}$ and $Q_{ac}$ commute  with the absolutely
continuous component $H^{ac}_0$ of $H_0$. 

Let $\Pi(H^{ac}_0)$ be a spectral representation of the absolutely
continuous part $H^{ac}_0$ of $H_0$, cf. Appendix \ref{C}. Since the components $\rho_{ac}$ and $Q_{ac}$
commute with $H^{ac}_0$, they are unitarily equivalent to
multiplication operators induced by some density and charge 
fiber matrices $\{\rho_{ac}(\gl)\}_{\gl\in\dR}$ and
$\{Q_{ac}(\gl)\}_{\gl\in\dR}$ in $\Pi(H^{ac}_0)$, respectively. In \cite{Pillet2007} it was
proved that the current $J^\cS_{\rho,Q}$ admits the representation
\be\la{0.7}
J^\cS_{\rho,Q} = \frac{1}{2\pi}\int_{\dR} d\gl\;
\tr\left\{\rho_{ac}(\gl)\left(Q_{ac}(\gl) - S(\gl)^*Q_{ac}(\gl)S(\gl)\right)\right\}.
\ee
The formula \eqref{0.7} can be called the abstract Landauer-B\"uttiker
formula. The formula (\ref{0.7}) is not identical with the traditional 
Landauer and B\"uttiker formula \eqref{0.2}. However, it was shown in
\cite{Pillet2007} that formula \eqref{0.2} follows from \eqref{0.7}.

The aim of the present paper is to extend the
representation \eqref{0.7} to situations where the operators $H$ and
$H_0$ might not be bounded from below. Using the intertwining property of the wave operator and the 
trace cyclicity,  one can rewrite the current $J^\cS_{\rho,Q}$ in the following form:
\be\la{0.9}
J^\cS_{\rho,Q} := -i\tr(W_-(H,H_0)(I + H^2_0)\rho
W_-(H,H_0)^*(H-i)^{-1}[H,Q](H+ i)^{-1}).
\ee
It turns out that \eqref{0.9} can be expressed in a different form 
using the Cayley transforms
\bed
U = (i - H)(i + H)^{-1}=e^{2i\arctan(H)} \quad \mbox{and} \quad 
U_0 = (i-H_0)(i+H_0)^{-1} =e^{2i\arctan(H_0)}
\eed
of $H$ and $H_0$, respectively. Under the condition that $V:=U-U_0=2i((i+H)^{-1}- (i+H_0)^{-1})$ is a
trace class operator  we have
\bed
\gO_\pm(U,U_0) := \slim_{n\to\pm\infty}U^nU^{-n}_0P^{ac}(U_0)=W_\pm(2\arctan(H),2\arctan(H_0))=W_{\pm}(H,H_0),
\eed
where in the last equality we used the invariance principle of wave
operators. Moreover, using the identity
\bed
-\frac{i}{2}U^*[U-U_0,Q]=(H-i)^{-1}[H,Q](H+ i)^{-1}
\eed
the current can be rewritten as
\be\label{mitlef1}
J^\cU_{\tilde{\rho},Q} := -\frac{1}{2}\tr(\gO_-(U,U_0)\tilde{\rho}U^*_0\gO_-(U,U_0)^*[V,Q]), \quad V := U - U_0,\quad 
\tilde{\rho}:= (1 + H^2_0)\rho,
\ee
where everything only depends on the unitary scattering 
system $\cU:=\{U,U_0\}$. Following Birman and Krein \cite{Krein1962,Birman1962}
we start with the abstract unitary scattering system 
$\cU:=\{U,U_0\}$ where $V=U-U_0$ is trace class operator, $\tilde{\rho}$ is 
an initial steady state and $Q$ a charge both commuting with
$U_0$.  Their restrictions to the absolutely continuous subspace 
of $U_0$ are denoted by $\tilde{\rho}_{ac}$ and $Q_{ac}$, respectively. 
Using a spectral representation of $U_0$, we denote by $\{\tilde{S}(\zeta)\}_{\zeta\in\dT}$,
$\{\tilde{\rho}_{ac}(\zeta)\}_{\zeta\in\dT}$ and 
$\{\tilde{Q}_{ac}(\zeta)\}_{\zeta\in\dT}$ the scattering, density and charge fiber matrices
of $S = \gO_+(U,U_0)^*\gO_-(U,U_0)$, $\tilde{\rho}_{ac}$ and $Q_{ac}$,
respectively. We also suppose that the singular continuous spectrum
$\gs_{sc}(U)$ of $U$ is empty (note that we allow $\gs_{sc}(U_0) \neq
\emptyset$).  Then it will be proven in Theorem \ref{II.9} and in
Corollary \ref{cor:EquilibriumFluxZero} that the current in \eqref{mitlef1} admits the representation
\be\label{mitlef2}
\begin{split}
J^\cU_{\tilde{\rho},Q} &= \frac{1}{4\pi}\int_\dT 
\tr\left\{\tilde{\rho}_{ac}(\zeta)
\left(\tilde{Q}_{ac}(\zeta) - \tilde{S}(\zeta)^*\tilde{Q}_{ac}(\zeta)\tilde{S}(\zeta)\right)\right\}d\nu(\zeta)\\
&= \frac{1}{4\pi}\int_\dT \tr\left\{
\left (\tilde{\rho}_{ac}(\zeta) - \tilde{S}(\zeta)\tilde{\rho}_{ac}(\zeta)\tilde{S}(\zeta)^*\right)\tilde{Q}_{ac}\right\}d\nu(\zeta),
\end{split}
\ee
 where $\nu$ is the Haar measure with  $\nu(\dT)=2\pi$. 
 
More importantly, from the second formula above we see that if 
$\tilde{\rho}$ is an equilibrium state, i.e. some non-negative
function $\tilde{f}(U_0)$, then it is a scalar multiplication with
$\tilde{f}(\zeta)$ on each fiber $\tilde\sh(\gl)$ and thus commutes with
$\tilde{S}(\zeta)$ almost everywhere. This shows that the current is
zero at equilibrium. Moreover, we can use this to renormalize 
the current by subtracting zero in the following way: 
\be\label{mitlef20}
J^\cU_{\tilde{\rho},Q} = \frac{1}{4\pi}\int_\dT 
\tr\left\{\left(\tilde{\rho}_{ac}(\zeta)-\tilde{f}(\zeta)I_{\tilde\sh(\zeta)}\right)\; 
\left(\tilde{Q}_{ac}(\zeta) - \tilde{S}(\zeta)^*\tilde{Q}_{ac}(\zeta)\tilde{S}(\zeta)\right)\right\}d\nu(\zeta).
\ee
Going back to the self-adjoint case via the Cayley transform, we have to change the torus with
the real line by the transformation $\zeta=e^{2i \arctan(\lambda)}$. 
Hence, replacing $\tilde{\rho}(e^{2i \arctan(\lambda)})$ by
$(1+\lambda^2)\rho(\lambda)$ and introducing 
$Q_{ac}(\gl):=\tilde{Q}_{ac}(e^{2i \arctan(\lambda)})$ and 
$S(\gl):=\tilde{S}(e^{2i \arctan(\lambda)})$ we obtain
\be\la{mitlef21}
J^\cS_{\rho,Q} = \frac{1}{2\pi}\int_{\dR} d\gl\;
\tr\left\{\left(\rho_{ac}(\gl)-f(\gl)I_{\sh(\gl)}\right)\;\left(Q_{ac}(\gl) - S(\gl)^*Q_{ac}(\gl)S(\gl)\right)\right\}.
\ee
This formula is very useful in the relativistic situation when 
$\rho_{ac}(\gl)$ can loose its decay in $\lambda$ at $-\infty$, 
as it happens with the Fermi-Dirac distribution. In that case we see 
that $\rho_{ac}(\gl)-f_{FD}(\gl)I_{\sh(\gl)}$ still decays 
exponentially at $\pm\infty$ and the current will be finite.
\\[-1mm] 

\noindent
Let us make the following remarks:

\begin{itemize}
\item Our main technical result is formula \eqref{mitlef2}, 
proved in Theorem \ref{II.9}. It can be seen as an abstract Landauer-B\"uttiker formula for unitary scattering systems. 

\item Formula \eqref{0.7} is proved in Theorem \ref{thm:LandBuett}, which is an extension of the result in
\cite{Nenciu2007}, where $V: = H - H_0 \in \sL_1(\sH)$ was assumed. 

\item Another result related to Theorem \ref{thm:LandBuett} was proven in 
\cite{Pillet2007} where the current was  defined through a regularization procedure. There the operators  
$H$ and $H_0$ were replaced by 
$H(1+\eta H)^{-N}$ and $H_0(1+\eta H_0)^{-N}$, respectively, 
and the limit $\eta \rightarrow +0$ was taken outside the trace. 
Using our approach via the Cayley transforms one gets a definition
of the current (see \eqref{mitlef1} or \eqref{0.9}) which avoids any regularization.  Since the Cayley transform does not require $H_0$ 
and $H$ to be bounded from below, it allows us to derive
Landauer-B\"uttiker type formulas for self-adjoint dilations of maximal dissipative Schr\"odinger operators
and Dirac operators with point interactions at zero, see Section \ref{IV}. 

\item 
Our result is stronger that that one of \cite{Pillet2007}. 
At first glance it seems to be that the condition
$(H + \theta)^{-N} - (H_0 + \theta)^{-N} \in \sL_1(\sH)$ assumed in \cite{Pillet2007} for some $N \in \dN$ and $\gth > 0$ is weaker than 
our condition $(i+H)^{-1}- (i+H_0)^{-1} \in \sL_1(\sH)$. Nevertheless, the result of \cite{Pillet2007} follows
from Theorem \ref{thm:LandBuett}. Indeed, let us assume for simplicity
that $H \ge I$ and $H_0 \ge I$ as well as $\gth = 0$. A straightforward computation shows that the representation
\be\la{1.16}
J^\cS_{\rho,Q} = -\frac{i}{N}\tr\left(W_-(H,H_0)\frac{I+
    H^{2N}_0}{H^{N-1}_0}\rho W_-(H,H_0)^*(H^N - i)^{-1}[H^N,Q](H^N+
  i)^{-1}\right)
\ee
is valid provide $(I + H^{N+1}_0)\rho$ is a bounded operator. 
Therefore, considering the trace class scattering system $\wh{\cS} = \left\{H^N,H^N_0\right\}$, 
we find 
\bed
J^\cS_{\rho,Q} = \tfrac{1}{N}J^{\wh \cS}_{\wh{\rho},Q}, \qquad \wh{\rho} := {H^{-(N-1)}_0}\rho,
\eed
where the invariance principle for wave
operators was taken into account. Finally, applying Theorem
\ref{thm:LandBuett} to $J^{\wh \cS}_{\wh{\rho},Q}$ we get a Landauer-B\"uttiker formula for the
scattering system $\wh{\cS} = \left\{H^N,H^N_0\right\}$ with respect
to a spectral representation of $(H^N_0)^{ac}$. However, from the
spectral representation of $(H^N_0)^{ac}$ one
easily obtains a spectral representation of $H^{ac}_0$ which
immediately implies the result of \cite{Pillet2007}.

\item 
\item 
We can extend Theorem 3.9 to some situations where $H$ and $H_0$ 
are not bounded from below and $(H+i)^{-1}-(H_0+i)^{-1}$ is not trace
class. Namely, if $0$ belongs to the resolvent set of both $H$ and
$H_0$, and if there exists an odd integer $N$ such that
$H^{-N}-H^{-N}_0$ is trace class, then the invariance principle can
still be applied and formula (1.16) (see also (1.12)) still makes sense. The general case remains open. 
\end{itemize}

\noindent
The paper is organized as follows. In Section \ref{II} we review some well known results related to 
non equilibrium steady states and currents, and extend them to the case of non-semibounded self-adjoint
operators $H_0$ and $H$. The main goal is to rigorously justify formula \eqref{mitlef1}. 

Section \ref{III} is devoted to the proof of the abstract Landauer-B\"uttiker formula \eqref{mitlef2}, 
at first for unitary operators, cf. Section \ref{sec:require},  and
then for self-adjoint operators, cf. Section \ref{III.2}. 

In Section \ref{IV} we give several examples. Finally, in order to make the paper
self-contained we have added Appendices \ref{A} and \ref{B}, \ref{C} on spectral representations of
unitary operators, and Appendix \ref{D}  on the
scattering matrix of unitary operators.

{\bf Notation:} By $\sH^{ac}(U)$ we denote the absolutely
continuous subspace of a unitary operator $U$ defined on
$\sH$. The projection from $\sH$ onto $\sh^{ac}(U)$ is denoted by
$P^{ac}(U)$. The corresponding absolutely continuous restriction of $U$ is denoted by
$U^{ac} := U\upharpoonright\sH^{ac}(U)$. The singular subspace of a
unitary operator $U$ is defined by $\sH^s(U) :=
\sH \ominus \sH^{ac}(U)$, the corresponding singular part by $U^s :=
U\upharpoonright\sH^s(U)$. 
A similar notation is used for self-adjoint operators. 

Furthermore the real axis and the unit circle are denoted by $\dR$, 
and $\dT$ respectively. The open unit disc is denoted by $\dD :=
\{\zeta \in \dC: |\zeta| < 1\}$.

\section{Steady states  and currents} \la{II}

Let $H_0$ be a self-adjoint operator and let $\rho$ be a steady state for $H_0$.
Furthermore, let us assume that at $t<0$ the system is described by the Hamiltonian $H_0$ and the steady state $\rho$. 
At $t= 0$ we switch on a coupling such that the system is now described by the Hamiltonian $H$.
The state $\rho(t)$ evolves according to the quantum Liouville equation
\bed
i\frac{d \rho}{d t} = [H,\rho(t)], \quad t > 0,\quad \rho(0)=\rho,
\eed
which has the weak solution
\bed
\rho(t) = e^{-itH}\rho e^{itH}, \quad t \ge 0.
\eed
The operator $\rho(t)$ is a density operator, but not a steady state for $H$. However, one can produce a steady state by taking an ergodic limit as in \cite{Pillet2007}. It turns out that Theorem 3.2 of \cite{Pillet2007} remains true even
if $H$ and $H_0$ are not semibounded; for completeness we formulate and prove below the result. 
\begin{proposition}
Let $H_0$ be a self-adjoint operator and let $\rho$ be a steady 
state of $H_0$. If $H$ is another self-adjoint operator 
such that $(H+i)^{-1}-(H_0+i)^{-1}$ is a trace class
operator and $\gs_{sc}(H) = \emptyset$, then the limit
\be\la{2.2a}
\rho_+ := \slim_{T\to\infty}\frac{1}{T}\int^T_0 \rho(t) dt
\ee
exists and is given by 
\be\la{2.1a}
\rho_+ = W_-(H,H_0)\rho W_-(H,H_0)^* + \sum_{\gl_k \in
  \gs_p(H)}E_H(\{\gl_k\})\rho E_H(\{\gl_k\})
\ee
where $E_H(\cdot)$ is the spectral measure of $H$ and $\gs_p(H)$
denotes the point spectrum of $H$, cf \cite[Theorem
3.2]{Pillet2007}. Moreover, $\rho_+$ is a steady state of $H$.
\end{proposition}
\begin{proof}
We use the representation
\bed
\rho(t) = e^{-itH}e^{itH_0}\rho e^{-itH_0}e^{itH}P^{ac}(H) +
e^{-itH}\rho e^{itH}P^{p}(H),
\quad t \ge 0,
\eed
where $P^p(H)$ denotes the projection onto the subspace spanned by the
eigenvectors of $H$. Notice that $P^p(H) = P^s(H)$ where $P^s(H)$ is
the projection onto the singular subspace of $H$. 
Since the resolvent difference is a trace class operator one gets
\bed
\slim_{T\to\infty}\frac{1}{T}\int^T_0 e^{-itH}e^{itH_0}\rho e^{-itH_0}e^{itH}P^{ac}(H)dt
= W_-(H,H_0)\rho W_-(H,H_0)^*.
\eed
Let $\gl_k \in \gs_p(H)$. We find 
\bed
e^{-itH}\rho e^{itH}E_H(\{\gl_k\}) = e^{-it(H-\gl_k)}\rho E_H(\{\gl_k\}), \quad t
\ge 0.
\eed
If $f = (H-\gl_k)g$, $g \in \dom(H)$, then
\bed
\frac{1}{T}\int^T_0 e^{-it(H-\gl_k)}f dt =
\frac{e^{-iT(H-\gl_k)} - I}{-iT}g
\eed
which yields
\bed
\lim_{T\to\infty}\frac{1}{T}\int^T_0 e^{-it(H-\gl_k)}f dt = 0
\eed
Since $\ran(H-\gl_k)$ is dense in $E_H(\dR \setminus\{\gl_k\})\sH$ we
verify that
\bed
\lim_{T\to\infty}\frac{1}{T}\int^T_0 e^{-it(H-\gl_k)}E_H(\dR\setminus\{\gl_k\})dt = 0
\eed
Finally, using the decomposition
\bed
\begin{split}
e^{-itH}\rho &e^{itH}E_H(\{\gl_k\}) =
e^{-it(H-\gl_k)}E_H(\dR\setminus\{\gl_k\})\rho E_H(\{\gl_k\})+\\
& 
E_H(\{\gl_k\})\rho E_H(\{\gl_k\}), \quad t \ge 0,
\end{split}
\eed
which proves
\bed
\slim_{T\to\infty}\frac{1}{T}\int^T_0 e^{-itH}\rho
e^{itH}E_H(\{\gl_k\})dt = E_H(\{\gl_k\})\rho E_H(\{\gl_k\}).
\eed
Using that we immediately prove \eqref{2.1a}.
\end{proof}

\vspace{0.5cm}

Formally, the current $J^\cS_{\rho,Q}$ is defined by 
\bed
J^\cS_{\rho,Q} = -\dE_{\rho_+}(i[H,Q]) = -i\tr(\rho_+[H,Q]),
\eed
where $\dE_{\rho_+}(\cdot)$ is the expectation value of an 
observable with respect to $\rho_+$. 
In general, the definition might be not correct because either the commutator
$[H,Q]$ is not well-defined or the product $\rho_+i[H,Q]$ is not a trace class
operator. To avoid such difficulties we set
\be\la{2.3x}
J^\cS_{\rho,Q}(\gd) := -\dE_{\rho_+}(iE_H(\gd)[H,Q]E_H(\gd)) = -i\tr(\rho_+E_H(\gd)[H,Q]E_H(\gd))
\ee
where $\gd$ is any bounded Borel set of $\dR$. 
Furthermore, $E_H(\gd)[H,Q]E_H(\gd)$ is a well defined trace class operator for any bounded Borel set $\gd$.
Indeed, using the representation
\be\la{2.4x}
E_H(\gd)[H,Q]E_H(\gd) = (H-i)E_H(\gd)K E_H(\gd)(H+i)
\ee
where
\begin{align}\la{2.5x}
K := & (H-i)^{-1}[H,Q](H+i)^{-1} =(H+i)(H-i)^{-1}[(H+i)^{-1},Q]\\
= & (I+2i(H-i)^{-1})[(H+i)^{-1}-(H_0+i)^{-1},Q]\nonumber 
\end{align}
is trace class. We get that $E_H(\gd)[H,Q]E_H(\gd)$ is a
trace class operator for every bounded Borel set $\gd$. We set
\bed\label{mitlef3}
J^\cS_{\rho,Q} := \lim_{\gd\to\dR}J^\cS_{\rho,Q}(\gd)
\eed
provided the limit exists. We show this now.
\begin{proposition}
Let $H_0$ be a self-adjoint operator and let $\rho$ be a steady 
state for $H_0$ and let $H$ be a self-adjoint operator. Further, let $Q$ be a charge for $H_0$. 
If the resolvent difference of $H$ and $H_0$ is a trace class operator, $\gs_{sc}(H) = \emptyset$ and $(I + H^2_0)\rho$ is 
a bounded operator, then the current $J^\cS_{\rho,Q}$ is well-defined and admits the representation
\eqref{0.9}.
\end{proposition}
\begin{proof}
Inserting \eqref{2.4x} into \eqref{2.3x} we get 
\bed
J^\cS_{\rho,Q}(\gd) := i\tr(\rho_+(H-i)E_H(\gd)KE_H(\gd)(H+i))
\eed
where $K$ is a trace class operator defined by \eqref{2.5x}.
Using \eqref{2.1a} we get
\begin{align*}
J^\cS_{\rho,Q}(\gd) &= 
-i\tr(W_-(H,H_0)\rho W_-(H,H_0)^*(H-i)E_H(\gd)KE_H(\gd)(H+i)) \nonumber \\
&-{i\sum}_{\gl_k\in\gs_p(H)\cap \gd}\tr(\rho E_H(\{\gl_k\})(H-i)K(H+i)E_H(\{\gl_k\})).
\end{align*}
Since $E_H(\{\gl_k\})KE_H(\{\gl_k\}) =0$
we find
\bed
J^\cS_{\rho,Q}(\gd) = 
-i\tr(W_-(H,H_0)(H^2_0 + I)\rho W_-(H,H_0)^*E_H(\gd)KE_H(\gd)),
\eed
where we have used that $(H^2_0 + I)\rho$ is a bounded operator. Then 
the limit in \eqref{mitlef3} exists and 
equals:
\bed 
J^\cS_{\rho,Q} = -i\tr(W_-(H,H_0)(H^2_0 + I)\rho W_-(H,H_0)^*K).
\eed
Note that \eqref{2.5x} coincides with \eqref{0.9}.
\end{proof}

\section{Landauer-B\"uttiker formula for unitary scattering systems}\la{III}

\subsection{Unitary operators}	\label{sec:require}

Let us recall that we consider two unitary operators $U$ and 
$U_0$ such that $U-U_0$ is trace class, and a bounded self-adjoint
operator $Q$ commuting with $U_0$ is called a charge. Thus any charge 
$Q$ is reduced by $\sH^{ac}(U_0)$ and $\sH^s(U_0)$. 
In other words, $Q$ admits the decomposition
$Q = Q_{ac} \oplus Q_s$ 
where $Q_{ac} := Q\upharpoonright\sH^{ac}(U_0)$ and $Q_s :=
Q\upharpoonright\sH^s(U_0)$. Notice that the restrictions $Q_{ac}$ and $Q_s$ 
might not be identical with the absolutely continuous and singular components $Q^{ac}$ and $Q^s$, respectively.

Let $\Pi(U^{ac}_0) = \{L^2(\dT,d\nu(\zeta),\sh(\zeta)),M,\Phi\}$ be a spectral
representation of $U^{ac}_0$, cf. Appendix \ref{A}.
Since $Q_{ac}$ commutes with $U^{ac}_0$
there is a measurable family $\{Q_{ac}(\zeta)\}_{\zeta\in\dT}$ of bounded
self-adjoint operators acting on $\sh(\zeta)$ such that
\bed
\nu-\sup_{\zeta \in \dT}\|Q_{ac}(\zeta)\|_{\cB(\sh(\zeta))}
=\|Q_{ac}\|_{\cB(\sH)}
\eed
and $Q_{ac} = \Phi^{-1}M_{Q_{ac}}\Phi$ where $M_{Q_{ac}}$ is the
multiplication operator induced by $\{Q_{ac}(\zeta)\}_{\zeta\in\dT}$ 
in $L^2(\dT,d\nu(\zeta),\sh(\zeta))$. 

A non-negative bounded self-adjoint operator $\rho$ commuting with $U_0$ is also called a density
operator and admits the decomposition $\rho = \rho_{ac} \oplus \rho_{s}$. 
The part $\rho_{ac}$ is unitarily equivalent to the multiplication
operator $M_{\rho_{ac}}$ induced by a measurable family $\{\rho_{ac}(\zeta)\}_{\zeta\in\dT}$ of
non-negative bounded operators acting on $\sh(\zeta)$ and
satisfying
$\nu-\sup_{\zeta\in\dT}\|\rho_{ac}(\zeta)\|_{\sh(\zeta)}
= \|\rho_{ac}\|_\sH$ in 
$L^2(\dT,d\nu(\zeta),\sh(\zeta))$.

Let $\cS = \{U,U_0\}$ be an $\sL_1$-scattering system.
Further, let $Q$ be a charge and let $\rho$ be a density operator. In
this case we define the current $J$ for $\cS$ by
\be\la{2.60}
J := -\frac{1}{2}\tr(\gO_-\rho U^*_0\gO^*_-[V,Q])
\ee
where $V=U-U_0$ is trace class and $[V,Q] = VQ - QV$. The 
main result of this section (see Proposition \ref{II.1}) will 
show that only the absolutely continuous restriction of $Q$ contributes to the current:
\be\label{2.12}
J = J_{ac} := -\frac{1}{2}\tr(\gO_-\rho U^*_0\gO^*_-[V,Q_{ac}]).
\ee
Before that, we need a series of lemmata.
\bl\label{II.2}
Let $U_0$ be a unitary operator on $\sH$ and let $Q$ be a charge.
Then $\sH$ admits an orthogonal decomposition
\bed
\sH = \bigoplus_{k\in\dN\sH_k}
\eed
reducing $U_0$ and $Q$ such that $U_k := U_0 \upharpoonright\sH_k$, $k \in \dN$,  has a
constant spectral multiplicity function and $Q_k :=
Q\upharpoonright\sH_k$ commutes with $U_k$, $k \in \dN$.
\el
\begin{proof}
Let $\Pi(U_0) =\{L^2(\dT,d\mu(\zeta),\sk(\zeta)),M,\Psi\}$ be a spectral
representation of $U_0$, cf Appendix \ref{A}, and
let ${\rm Mult}(\zeta) := \dim(\sk(\zeta))$ be the spectral multiplicity function of $U_0$. We set $\gD_1 :=
\{\zeta \in \dT: {\rm Mult}(\zeta) =\infty\}$ and $\gD_k :=
\{\zeta \in \dT: {\rm Mult}(\zeta) =k-1\}$ if $k \geq 2$. Let  $E_0(\cdot)$ be 
the spectral measure of $U_0$. We set $\sH_k :=
E_0(\gD_k)\sH$. Obviously, each subspace $\sH_k$ reduces $U_0$ and
$Q$. Moreover, the unitary operators $U_k$ defined on $\sH_k$ are of
constant spectral multiplicity.  
\end{proof} 
Next we are going to show that $Q$ can be approximated by a sequence of
self-adjoint operators with pure point spectrum.
\bl\label{II.3}
Let $U_0$ be a unitary operator on $\sH$ of constant spectral
multiplicity and let $Q$ be a  charge. 
Then there is a sequence $\{Q_m\}_{m\in\dN}$ of charges with pure
point spectrum satisfying $\slim_{m\to\infty}Q_m = Q$ and 
$\|Q_m\|_{\sH} \le \|Q\|_{\sH} +1$.
\el
\begin{proof}
Since $U_0$ is of constant spectral
multiplicity $U_0$ admits the spectral representation $\Pi(U_0) := \{L^2(\dT,d\mu(\zeta),\sk),M,\Psi\}$
where $\sk$ is independent from $\zeta \in \dT$. 
If $Q$ is a charge, then there is a measurable family
$\{Q(\zeta)\}_{\zeta\in\dT}$ of bounded self-adjoint operators
satisfying
${\mu-\sup_{\zeta\in\dT}}\|Q(\zeta)\|_{\sk} = \|Q\|_\sH$ such that $Q$
is unitarily equivalent to the multiplication operator $M_Q$ in $L^2(\dT,\mu(\zeta),\sk)$.

Since $\{Q(\zeta)\}_{\zeta\in\dT}$ is a measurable family of self-adjoint
operators there is a sequence $\{\wt Q_m(\zeta)\}_{\zeta\in\dT}$ of simple functions
such that 
\be\label{2.15}
\slim_{m\to\infty}\wt Q_m(\zeta) = Q(\zeta)
\ee
for a.e. $\zeta \in \dT$ with respect to $\mu$. We recall that $\wt Q_m(\cdot)$ is simple if
it admits the representation 
\bed
\wt Q_m(\zeta) = \sum_l \chi_{\gd_{ml}}(\zeta)\wt Q_{ml}, \quad \zeta \in
\dT, \quad \wt Q_{ml} = \wt Q^*_{ml} \in \sB(\sk),
\eed
where $\{\gd_{ml}\}$ are disjoint
Borel subsets of $\dT$ satisfying $\bigcup_{l}\gd_{ml} = \dT$ for each
$m \in \dN$ and $\sum_l$ is finite. Without loss of generality we can assume that the condition
\bed
\|\wt Q_m(\zeta)\|_{\sk} \le \mu-\sup_{\eta\in\dT}\|\wt Q_m(\eta)\|_{\sk}
\eed
is satisfied for each $m \in \dN$.

By the v.~Neumann theorem \cite[Theorem X.2.1]{Ka1995} for each
self-adjoint operator $\wt Q_{ml}$
there is a self-adjoint Hilbert-Schmidt operator $D_{ml}$ such that
$\|D_{ml}\|_{\sL_2} \le \frac{1}{m}$ and $Q_{ml} := \wt Q_{ml} + D_{ml}$ is pure
point. Setting 
\bed
Q_m(\zeta) = \sum_l \chi_{\gd_{ml}}(\zeta) Q_{ml}, \quad \zeta \in
\dT, \quad Q_{ml} = Q^*_{ml} \in \sB(\sk),
\eed
one easily verifies that
\bed
\slim_{m\to\infty}Q_m(\zeta) = Q(\zeta)
\eed
for a.e. $\zeta \in \dT$ with respect to $\mu$. We note that $\slim_{m\to\infty}M_{Q_m} = M_Q$. Moreover, the
spectrum of $M_{Q_m}$ is pure point for each $m \in \dN$.
Setting $Q_m :=  \Psi^{-1} M_{Q_m}\Psi$
we find that $\slim_{m\to\infty}Q_m = Q$. Moreover, each operator
$Q_m$ commutes with $U_0$. 
\end{proof}
\bl\label{II.4}
Let $U_0$ be a purely singular unitary operator (i.e. the absolutely continuous component is absent) 
on the separable Hilbert space $\sH$. Then there is a sequence
$\{U_n\}_{n\in\dN}$ of unitary operators with pure point spectrum such that $U_0 - U_n \in \sL_1(\sH)$, $n\in \dN$, and
$\lim_{n\to\infty}\|U_0 - U_n\|_{\sL_1} = 0$.
\el
\begin{proof}
Let us assume that $\ker(U_0 + I) = \{0\}$. We introduce the self-adjoint operator
\bed
H_0 = i(I - U_0)(I + U_0)^{-1}.
\eed
Since $U_0$ is singular the self-adjoint operator $H_0$ is also
singular. By Lemma 2 of \cite{CP1976} for each $n \in \dN$ there is a self-adjoint trace class operator $D_n$ 
satisfying $\|D_n\|_{\sL_1} < \frac{1}{n}$ such that $\wt H_n := H_0 +
D_n$ is pure point. Hence, the unitary operators
\bed
U_n := (i - \wt H_n)(i + \wt H_n)^{-1}, \quad n \in \dN,
\eed
have pure point spectrum. Since
\bed
U_0 - U_n = 2i(i + \wt H_n)^{-1}D_n(i + H_0)^{-1}, \quad n\in \dN,
\eed
we get
\bed
\|U_0 - U_n\|_{\sL_1} \le 2\|D_n\|_{\sL_1} < \frac{2}{n}, \quad n \in \dN,
\eed
which yields $\slim_{n\to\infty}\|U_0 -U_n\|_{\sL_1} = 0$.

If the condition $\ker(I + U_0) = 0$ is not satisfied, then the unitary
operator admits the decomposition $U_0 = U'_0 \oplus U''_0$ where
$U'_0 = U_0\upharpoonright \sH'$, $\sH' := \ker(I + U_0)^\perp$, and
$U''_0 = U_0\upharpoonright\sH'' = -I_{\sH''}$, $\sH'' := \ker(I + U_0)$. One easily verifies that $\ker(I + U'_0) = \{0\}$. 
Hence the construction above can be applied. 
That means, there is a sequence $\{U'_n\}_{n\in\dN}$ of unitary operators 
with simple pure point spectrum on $\sH'$ such that $U'_0 - U'_n \in \sL_1(\sH')$, $n \in
\dN$, and  $\lim_{n\to\infty}\|U'_0 -U'_n\|_{\sL_1} = 0$. 

On the Hilbert space $\sH''$ we choose  $U''_n = -I$, $n \in \dN$. Setting $U_n := U'_n \oplus U''_n$, $n \in \dN$, we complete the
proof.
\end{proof}
\begin{proposition}\label{II.5}
Let $U_0$ be a purely singular unitary operator 
and let $Q$ be a charge, both acting on the separable Hilbert
space $\sH$. Then there is  a sequence of unitary operators
$\{\wt U_m\}_{m\in \dN}$ and a sequence of
bounded self-adjoint operators $\{Q_m\}_{m\in\dN}$ both with pure
point spectrum such that $[Q_m,\wt U_m]=0$ and $U_0 - \wt U_m
\in \sL_1$ for all $m \in \dN$ satisfying
\bed
\lim_{m\to\infty}\|U_0 - \wt U_m\|_{\sL_1} = 0 \quad \mbox{and} \quad  
Q = \slim_{m\to\infty}\wt Q_m.
\eed
\end{proposition}
\begin{proof}
By Lemma \ref{II.2} we find a decomposition
\bed
U_0 = \bigoplus_{k\in\dN} U_k \quad \mbox{and} \quad Q = \bigoplus_{k\in\dN} Q_k
\eed
where $U_k$ is of constant spectral multiplicity and $Q_k$ are bounded self-adjoint operators
commuting with $U_k$ such that $\sup_{k\in\dN}\|Q_k\|_{\sH_k} =
\|Q\|_{\sH}$. 

By Lemma \ref{II.3} for each $k\in\dN$
there is a sequence $\{Q_{km}\}_{m\in\dN}$ of bounded self-adjoint operators with pure point spectrum commuting with $U_k$
such that $\|Q_{km}\|_{\sH_k} \le \|Q_k\|_{\sH} + 1$ for each $m \in
\dN$ and $Q_k = \slim_{m\to\infty}Q_{km}$. The operators $Q_{km}$
admit the representation
\bed
Q_{km} = \sum_{l\in\dN} \gl_{kml}P_{kml}
\eed
where $P_{kml}$ are eigenprojections of $Q_{kml}$ in $\sH_k$. Since
$U_k$ commutes with $Q_{km}$ the eigenprojections $P_{kml}$ commute
with $U_k$. We set $U_{kml} := U_k\upharpoonright\sH_{kml}$ where
$\sH_{kml} := P_{kml}\sH_k$. Notice that
\bed
U_{km} = \bigoplus_{l\in\dN} U_{kml}.
\eed
The unitary operators $U_{kml}$ are singular but their spectral
multiplicity might be not constant. 

By Lemma \ref{II.3} for each $k,m,l \in \dN$ there is a unitary
operator $\wt U_{kml}$ on $\sH_{kml}$ such that the spectrum of
$U_{kml}$ is pure point, $U_{kml} - \wt U_{kml} \in \sL_1(\sH_{kml})$
and 
\bed
\|U_{kml} - \wt U_{kml}\| \le \frac{1}{(k+m+l)^3}.
\eed
Obviously, $\wt U_{kml}$ commutes with $P_{kml}$. 
Setting
\bed
\wt U_{km} := \bigoplus_{l\in\dN}\wt U_{kml}
\eed
we get a unitary operator on $\sH_k$ with pure point spectrum which
commutes with $Q_{km}$. Moreover, we have
\bed
\|U_{km} - \wt U_{km}\|_{\sL_1} \le \sum_{l\in\dN}\frac{1}{(k+m+l)^3}.
\eed
Finally, setting
\bed
\wt U_m := \bigoplus_{k\in\dN} \wt U_{km} \quad \mbox{and} \quad Q_m :=
\bigoplus_{k\in\dN} Q_{km}
\eed
we define a unitary and a self-adjoint operator on $\sH$. Obviously, $\wt
U_m$ and $Q_m$ commute for each $m \in \dN$ and they are pure point.
Since 
\bed
\|U_0 - \wt U_m\|_{\sL_1} \le
\sum_{k\in\dN}\sum_{l\in\dN}\frac{1}{(m+k+l)^3}
\eed
we have $U_0 - \wt U_m \in\sL_1(\sH)$ for each $m \in \dN$ and
$\lim_{m\to\infty}\|U_0 - \wt U_m\|_{\sL_1} = 0$. We recall 
that $\slim_{m\to\infty}Q_m = Q$ by Lemma \ref{II.2}.
\end{proof}
\begin{proposition}\label{II.1}
Let $\cS = \{U,U_0\}$ be a $\sL_1$-scattering system.
Further, let $Q$ be a charge and $\rho$ be a density operator. 
If $U-U_0$ is trace class, then $J = J_{ac}$ (see \eqref{2.12}), i.e. the pure point and singular continuous spectral subspaces of $U_0$ do not contribute to the steady current.
\end{proposition}
\begin{proof}
Using the decompositions $U_0  = U^{ac}_0 \oplus U^s_0$  and $Q = Q_{ac}
\oplus Q_s$ we have:
\bed
J =- \frac{1}{2}\tr(\gO_-\rho U^*_0\gO^*_-[V,Q_{ac}]) -\frac{1}{2}\tr(\gO_-\rho U^*_0\gO_-^*[V,Q_s]).
\eed
We are going to show that $J_s := -\frac{1}{2}\tr(\gO_-\rho U^*_0\gO_-^*[V,Q_s] )= 0$.

Let us first assume that the spectra of $U^s_0$ and $Q_s$ are pure point. Hence $U^s_0$ and $Q_s$ admit the representations
\bed
U^s_0 = \sum_{n\in\dN}\zeta_nP_n \quad \mbox{and} \quad
Q_s = \sum_{l \in\dN}q_l Q_l
\eed
where $\zeta_n \in \dT$, $q_l \in \dR$ and $P_n$, $Q_l$ are eigenprojections of $U^s_0$ and $Q_s$, respectively. Since $U^s_0$ and
$Q_s$ commute, then their eigenprojections $P_n$ and $Q_l$ also commute. We set $Q_{nl} := P_nQ_l$, which define some orthogonal projections.  We have the representation
\bed
U^s = \sum_{n,l\in\dN}\zeta_{nl}Q_{nl} \quad \mbox{and} \quad
Q_s = \sum_{n,l\in\dN}q_{nl}Q_{nl}
\eed
where $\zeta_{nl} \in \dT$ and $q_{nl} \in \dR$. Notice that $\sum_{n,l\in\dN}Q_{nl} = P^s(U_0)$. 
Without loss of generality we can assume that $Q_{nl}$ are one dimensional orthogonal projections. 
Because the series $\sum_{n,l\in\dN}\zeta_{nl}Q_{nl} [V,Q_{nl}]$ converges in the trace class norm to $[V,Q_s]$, we can write:
\bed
J^s =  -\frac{1}{2}\sum_{n,l\in\dN}q_{nl}\tr(\gO_-\rho U^*_0\gO_-^*[V,Q_{nl}]).
\eed
Now we can undo each commutator and write:
\bed
\tr(\gO_-\rho\gO_-^*[V,Q_{nl}]) = \tr(\gO_-\rho U^*_0\gO_-^* UQ_{nl}) - \tr(\gO_-\rho U^*_0\gO_-^* Q_{nl}U).
\eed
Using trace cyclicity we have $\tr(\gO_-\rho U^*_0\gO_-^* Q_{nl}U)=\tr(U\gO_-\rho U^*_0\gO_-^* Q_{nl})$, and then because 
$U$ commutes with $\gO_-\rho U^*_0\gO_-^*$ due to the intertwining property of the wave operator, we can put $U$ at the left of  $Q_{nl}$. Hence $J_s = 0$.

If $U^s$ and $Q_s$ are not pure point, then in accordance with Proposition \ref{II.5} there is a sequence $\{U^s_m\}_{m\in\dN}$ of
pure point  unitary operators acting on $\sH^s(U_0)$ and a sequence $\{Q_{s,m}\}_{m\in\dN}$ of bounded self-adjoint operators 
with pure point spectrum acting on $\sH^s(U_0)$ such that $[U^s_m,Q_{s,m}] = 0$ and $U^s_0 - U^s_n \in \sL_1(\sH^s(U_0))$ for $m \in \dN$ as well as 
$\lim_{m\to\infty}\|U^s_0 - U^s_m\|_{\sL_1} = 0$ and $\slim_{m\to\infty}Q_m = Q_s$.

We set
\bed
U_m := U^{ac}_0 \oplus U^s_m
\quad \mbox{and} \quad Q_m := Q_{ac} \oplus Q_{s,m},
\quad m \in \dN.
\eed
We have $[U_m,Q_m] = 0$ and $U_0 - U_m \in \sL_1(\sH)$ for $m\in\dN$ as well as $\lim_{m\to\infty}\|U_0- U_m\|_{\sL_1} = 0$
and $\slim_{m\to\infty}Q_m = Q$. Since $U - U_m = U - U_0 + U_0 - U_m \in \sL_1(\sH)$ the wave operators
\bed
\gO_\pm(U,U_m) = \slim_{n\to\pm\infty}U^nU^{-n}_mP^{ac}(U_m)
\eed
exist for each $m \in \dN$. However, we have $\gO_\pm = \gO_\pm(U,U_m)$ for each $m \in \dN$ since $U^{ac}_m = U^{ac}_0$.   Let
\bed
J_m := -\frac{1}{2}\tr(\gO_-(U,U_m)\rho_{ac}U^{*}_0\gO_-(U,U_m)^*[V_m,Q_m]), \quad m \in
\dN,
\eed
where $V_m := U - U_m$. We note that $J_m = (J_m)_{ac} + (J_m)_s$ where
\bead
(J_m)_{ac} & := &- \frac{1}{2}\tr(\gO_-(U,U_m)\rho_{ac}U^{*}_0\gO_\pm(U,U_m)^*[V_m,Q_{ac}])\\
(J_m)_s & := &- \frac{1}{2}\tr(\gO_-(U,U_m)\rho_{ac}U^{*}_0\gO_\pm(U,U_m)^*[V_m,Q_{s,m}]).
\eead
Since $U^s_m$ and $Q_{s,m}$ are pure point we get by the considerations
above that $(J_m)_s = 0$ for each $m \in \dN$. Hence $J_m = (J_m)_{ac}$,
$m \in \dN$. 

Furthermore, using
$\gO_\pm = \gO_\pm(U,U_m)$ and $U^{ac}_0 = U^{ac}_m$ we find
\bed
J_m = (J_m)_{ac} = -\frac{1}{2}\tr(\gO_-\rho_{ac}U^{*}_0\gO^*_-[V_m,Q_{ac}]),\quad m \in
\dN.
\eed
Since $\lim_{m\to\infty}\|U_0 - U_m\|_{\sL_1} = 0$ and
$\slim_{m\to\infty}Q_m = Q$ we find $\lim_{m\to\infty}J_m  = J$ and 
$\lim_{m\to\infty}(J_m)_{ac} = J_{ac}$ which yields $J = J_{ac}$. 
\end{proof}
\bl\label{II.9a}
Let $\{U,U_0\}$ be a $\sL_1$-scattering system. With the notation introduced in \eqref{mitlef23}, let
\bed
J(r) := -\frac{1}{2}\tr(\gO_-(r)\rho U^*_0\gO_-(r)^*[V,Q_{ac}]), \quad r \in [0,1).
\eed
If $\gs_s(U) = \emptyset$, then $J =  \lim_{r\uparrow 1}J(r)$. 
\el
\begin{proof}
We set
\bed
J^{ac}(r) := -\frac{1}{2}\tr(\gO_-(r)\rho U^*_0\gO_-(r)^*P^{ac}(U)[V,Q_{ac}])
\eed
and
\bed
J^s(r) :=- \frac{1}{2}\tr(\gO_-(r)\rho U^*_0\gO_-(r)^*P^s(U)[V,Q_{ac}]).
\eed
Since $\gO_-^* = \slim_{r\uparrow 1}\gO_-(r)^*P^{ac}(U)$ one easily verifies that $J = \lim_{r\uparrow}J^{ac}(r)$.

Let us show that $\lim_{r\uparrow 1}J^s(r) = 0$.
To this end we verify that
\bed
\slim_{r\uparrow 1}\gO_-(r)^*P^s(U) = 0.
\eed
Let $\varphi_k$, $\|\varphi_k\| = 1$, be an eigenvector of $U$ corresponding to the eigenvalue $\xi_k \in \dT$.
One gets
\bed
\gO_-(r)^*\varphi_k = (1-r)P^{ac}_0\sum_{n\in\dN}r^k U^{-n}_0U^n\varphi_k = (1-r)P^{ac}_0\sum_{n\in\dN}r^k U^{-n}_0\xi^n_k\varphi_k. 
\eed
Hence
\bed
\gO_-(r)^*\varphi_k = P^{ac}_0\frac{1-r}{I - U^*_0\xi_k}\varphi_k = (1-r)\int_\dT \frac{1}{1 - r\overline{\zeta}\xi}dE^{ac}_0(\zeta)\varphi_k.
\eed
We introduce the Borel subset $\gD^N_k$ of $\dT$ defined by
\be\label{2.100a}
\gD^N_k := \left\{\zeta \in \dT: \frac{d(E^{ac}_0(\zeta)\varphi_k,\varphi_k)}{d\nu(\zeta)} \le N\right\}.
\ee
It is not hard to see that $\slim_{N\to\infty}E^{ac}_0(\dT \setminus \gD^N_k) = 0$. By the decomposition
\bead
\gO_-(r)^*\varphi_k & = & (1-r)\int_{\gD^N_k} \frac{1}{1 -
  r\overline{\zeta}\xi_k}dE^{ac}_0(\zeta)\varphi_k + \\
& & (1-r)\int_{\dT\setminus\gD^N_k} \frac{1}{1 - r\overline{\zeta}\xi}dE^{ac}_0(\zeta)\varphi_k
\eead
we find
\bead
\|\gO_-(r)^*\varphi_k\|^2 & = & 
\frac{1-r}{1+r}\int_{\gD^N_k}\frac{1 - r^2}{|1 -
  r\overline{\zeta}\xi_k|^2}\frac{d(E^{ac}_0(\zeta)\varphi_k,\varphi_k)}{d\nu(\zeta)} +\\
& &
(1-r)^2\int_{\dN\setminus\gD^N_k}\frac{1}{|1 - r\overline{\zeta}\xi_k|^2}\frac{d(E^{ac}_0(\zeta)\varphi_k,\varphi_k)}{d\nu(\zeta)}.
\nonumber
\eead
Taking into account \eqref{2.100a} we find the estimate
\bed
\|\gO_-(r)^*\varphi_k\|^2 \le 2\pi N \frac{1-r}{1+r} + 
(1-r)^2\int_{\dN\setminus\gD^N_k}\frac{1}{|1 - r\overline{\zeta}\xi_k|^2}\frac{d(E^{ac}_0(\zeta)\varphi_k,\varphi_k)}{d\nu(\zeta)}.
\eed
Using $\frac{(1-r)^2}{|1 - r\overline{\zeta}\xi_k|^2} \le 1$ we get
\bed
\|\gO_-(r)^*\varphi_k\|^2 \le 2\pi N \frac{1-r}{1+r} + (E^{ac}_0(\dT \setminus \gD^N_k)\varphi_k,\varphi_k). 
\eed
For each $\varepsilon > 0$ there is $N_0$ such that 
$(E^{ac}_0(\dT \setminus \gD^N_k)\varphi_k,\varphi_k) < \frac{\varepsilon}{2}$ for $N > N_0$. Fixing such a $N$ there is 
$r_0 < 1$ such that for $r \in (r_0,1)$ one has $2\pi N \frac{1-r}{1+r} < \frac{\varepsilon}{2}$. 
\bed
\|\gO_-(r)^*\varphi_k\|^2 \le  \varepsilon.
\eed
Hence $\lim_{r\uparrow 1}\|\gO_-(r)^*\varphi_k\|^2 = 0$. From the above considerations we get
$\lim_{r\uparrow 1}\gO^*_-(r)f = 0$ provided $f = \sum_k c_k f_k$,
$c_k \in \dC$, is a finite sum of eigenvectors of $U$. However, the
set of finite sums of eigenvectors  of $U$ is dense in $\sH^s(U)$
which yields $\slim_{r\uparrow 1}\gO^*_-(r)P^s(U) = 0$. 
Using $\slim_{r\uparrow 1}\gO_-(r) = \gO_-$ and the compactness of $V$ we immediately get that $\lim_{r\uparrow 1}J^s(r) = 0$.
\end{proof}

Using the results above we are now going to prove a Landauer-B\"uttiker formula for unitary operators

\bt\label{II.9}
Let $\cS = \{U,U_0\}$ be a $\sL_1$-scattering system. Further let $Q_0$ be a charge and let $\rho$ be a density operator.
If $\gs_{sc}(U) = \emptyset$, then 
\be\label{2.90}
J  = \frac{1}{4\pi}\int_\dT \tr\left\{\rho_{ac}(\zeta)[Q_{ac}(\zeta) - S(\zeta)^*Q_{ac}(\zeta)S(\zeta)]\right\}d\nu(\zeta)
\ee
where $S(\zeta)$ is the scattering matrix of the scattering system $\cS$.
\et
\begin{proof}
Let us introduce the approximate current by 
\bed
J(r,\varepsilon) :=-\frac{1}{2}\tr(\gO_-(r)\rho_{ac}^\varepsilon U^*_0\gO_-(r)^*[V,Q_{ac}]), \quad 0 \le r < 1,
\eed
where
\be\la{2.111}
\rho^\varepsilon_{ac} := E^{ac}_0(\gD_*(\varepsilon))\rho, \qquad
\varepsilon \ge 0,
\ee
and $\gD_*(\varepsilon) \subseteq \dT$ satisfying $\nu(\gD_*(\varepsilon)) <
\varepsilon$ and \eqref{2.55a}. Notice that $\rho^\varepsilon_{ac}$ is
also a density operator. By Lemma \ref{II.9a} we immediately get that
$\lim_{r\uparrow1}J(r,\varepsilon) = J(\varepsilon)$ where
\bed
J(\varepsilon) := -\frac{1}{2}\tr(\gO_-\rho^\varepsilon_{ac} U^*_0\gO_-^*[V,Q_{ac}]).
\eed
Furthermore, we note that
\be\la{2.114}
J = \lim_{\varepsilon \to +0}J(\varepsilon) =  \lim_{\varepsilon \to
  +0}\lim_{r\uparrow1} J(r,\varepsilon)
\ee
where $J$ is given by \eqref{2.60}. We set 
\bed
\begin{matrix}
J_1(\varepsilon) & := & \tr(\rho^\varepsilon_{ac} \gO^*_-V Q_{ac}\gO_-U^*_0),\\[1ex]
J_2(\varepsilon) & := & \tr(\rho^\varepsilon_{ac} U^*_0\gO^*_-Q_{ac}V\gO_-)
\end{matrix}
\eed
and 
\bed
\begin{matrix} 
J_1(r,\varepsilon) & := & \tr(\rho^\varepsilon_{ac}\gO_-(r)^*V Q_{ac}\gO_-(r)U^*_0), \\[1ex]
J_2(r,\varepsilon) & := & \tr(\rho^\varepsilon_{ac}\gO_-(r)^*Q_{ac}U_0V\gO_-(r)),
\end{matrix}
 \qquad 0 \le r < 1.
\eed
Notice that 
\be\label{2.91}
\begin{matrix}
-2J(\varepsilon) & = & J_1(\varepsilon) - J_2(\varepsilon)\\[1ex]
-2 J(r,\varepsilon) & = & J_1(r,\varepsilon) - J_2(r,\varepsilon),
\end{matrix}
\ee
$0 \le r < 1$. Setting $K(r) := \gO_-(r)^*V$, $0 \le r < 1$, we get 
\bed
J_1(r,\varepsilon) = \tr(\rho^\varepsilon_{ac} K(r) Q_{ac}\gO_-U^*_0),
\eed
Using $V = -U_0V^*U$ we obtain
which yields 
\be\label{2.96}
J_2(r,\varepsilon) := -\tr(\rho^\varepsilon_{ac}\gO_-(r)^*Q_{ac}U_0 K(r)^*).
\ee
At first, we are going to calculate $K(r)Q_{ac}\gO_-(r)U^*_0$. From \eqref{2.53}
we get
\bed
K(r)Q_{ac}\gO_-(r)U^*_0 = K(r)Q_{ac}\left\{P^{ac}_0 + r\int_\dT \frac{1}{I -r\zeta U^*}V^*U_0 dE_0^{ac}(\zeta)\right\}U^*_0
\eed
where we have used $U^*V = - V^*U_0$ which leads to
\bed
K(r)Q_{ac}\gO_-(r)U^*_0 = K(r)Q_{ac}\left\{U^*_0P^{ac}_0 + r\int_\dT \frac{1}{I -r\zeta U^*}V^*dE_0^{ac}(\zeta)\right\}.
\eed
Setting 
\be\label{2.100}
\Xi(r) := r\int_\dT \frac{1}{I -r\zeta U^*}V^*dE_0^{ac}(\zeta)
\ee
we get
\bed
K(r)Q_{ac}\gO_-(r)U^*_0 = K(r)Q_{ac}U^*_0P^{ac}_0 + K(r)Q_{ac}\Xi(r)
\eed
and
\bed
J_1(r,\epsilon) = \tr(\rho^\varepsilon_{ac} K(r)Q_{ac}U^*_0) + \tr(\rho^\varepsilon_{ac} K(r)Q_{ac}\Xi(r)).
\eed
Using the unitary operator $\Phi$ and \eqref{2.87} we find
\bed
(\Phi K(r)Q_{ac}U^*_0 \Phi^{-1} \wh f)(\zeta) = \int_\dT K(r;\zeta,\zeta')
Q_{ac}(\zeta') \overline{\zeta'}\wh f(\zeta')d\nu(\zeta'),
\eed
$\wh f \in L^2(\dT,d\nu(\zeta),\sh(\zeta))$. By the resolvent formula one has the identity
\bed
(I - \xi U^*)^{-1} = (I - \xi U^*_0)^{-1}\left\{I + \zeta V^*(I - \xi U^*)^{-1}\right\}, \quad \xi \in \dD.
\eed
Multiplying on the right by $V^*$ we get
\bed
(I - \xi U^*)^{-1}V^* = (I - \xi U^*_0)^{-1}\left\{V^* + \xi V^*(I - \xi U^*)^{-1}V^*\right\}, \quad \xi \in \dD,
\eed
which yields
\bed
(I - \xi U^*)^{-1}V^* = (I - \xi U^*_0)^{-1}CZ(\xi)C, \quad \xi \in \dD.
\eed
Using that we obtain
\bed
\Xi(r) =
r\int_\dT (I - r\zeta' U^*_0)^{-1}CZ(r \zeta')CdE_0^{ac}(\zeta')
\eed
which yields
\be\label{2.107}
\Xi(r) = 
r \int_\dT E^{ac}_0(d\xi)C \int_\dT (I - r\zeta' \overline{\xi})^{-1}Z(r \zeta')CdE_0^{ac}(\zeta').
\ee
Applying the map $\Phi$ one gets
\bed
(\Phi\Xi(r)\Phi^{-1}\wh f)(\xi) =
r \sqrt{Y(\xi)} \int_\dT (I - r\zeta' \overline{\xi})^{-1}Z(r \zeta')\sqrt{Y(\zeta')}\wh f(\zeta') d\nu(\zeta').
\eed
or
\bed
(\Phi\Xi(r)\Phi^{-1}\wh f)(\xi) =
r \int_\dT (I - r\zeta' \overline{\xi})^{-1}K(r;\zeta',\xi)^*\wh f(\zeta') d\nu(\zeta'),
\eed
$\wh f \in L^2(\dT,d\nu(\zeta),\sh(\zeta))$. Using
\bed
(\Phi K(r)Q_{ac}\Xi(r)\Phi^{-1}\wh f)(\xi) = (\Phi K(r)\Phi^{-1}\Phi Q_{ac}\Phi^{-1}\Phi\Xi(r)\Phi^{-1}\wh f)(\xi)
\eed
and \eqref{2.87} we find
\bead
\lefteqn{
(\Phi K(r)Q_{ac}\Xi(r)\Phi^{-1}\wh f)(\zeta) =}\\
& &
r\int_\dT d\nu(\xi) K(r;\zeta,\xi)Q_{ac}(\xi)\int_\dT d\nu(\zeta')(I - r\zeta' \overline{\xi})^{-1}K(r;\zeta',\xi)^*\wh f(\zeta').
\nonumber
\eead
Setting
\bea\label{2.112}
M(r;\zeta,\xi,\zeta') & := & K(r;\zeta,\xi)Q_{ac}(\xi)K(r;\zeta',\xi)^*\\
                      & = & \sqrt{Y(\zeta)}Z(r\zeta)^*\sqrt{Y(\xi)}Q_{ac}(\xi)\sqrt{Y(\xi)}Z(r\zeta')\sqrt{Y(\zeta')}\nonumber\\
& = & X_*(r;\zeta)\sqrt{Y(\xi)}Q_{ac}(\xi)\sqrt{Y(\xi)}X_*(r;\overline{\zeta'})^*
\nonumber
\eea
we find
\bed
(\Phi K(r)Q_{ac}\Xi(r)\Phi^{-1}\wh f)(\zeta) =
r\int_\dT d\nu(\xi)\int_\dT d\nu(\zeta')  \frac{M(r;\zeta,\xi,\zeta')}{I - r\zeta' \overline{\xi}}\wh f(\zeta')
\eed
where $X_*(r;\zeta)$ is defined by \eqref{2.55}.
Notice that 
\bed
M(r;\zeta,\xi,\zeta')^* = M(r;\zeta',\xi,\zeta).
\eed
Summing up we obtain
\bea\la{2.115}
J_1(r,\varepsilon) = \int_\dT
d\nu(\zeta)\overline{\zeta}\tr(\rho^\varepsilon_{ac}(\zeta)K(r;\zeta,\zeta)Q_{ac}(\zeta))
+
r \int_{\dT^2} d\nu(\zeta) d\nu(\xi)
\tr\left(\rho^\varepsilon_{ac}(\zeta)\frac{M(r;\zeta,\xi,\zeta)}{I - r\zeta
  \overline{\xi}}\right).
\eea
We are going to calculate $J_2(r,\varepsilon)$. From \eqref{2.60a} we get
\bed
\gO_-(r)^*Q_{ac}U_0K(r)^* = \left\{P^{ac}_0 + r\int_\dR
  dE^{ac}_0(\zeta)V\frac{\overline{\zeta}}{I -
    r\overline{\zeta}U}\right\}Q_{ac}U_0K(r)^*
\eed
or
\bed
\gO_-(r)^*Q_{ac}U_0K(r)^* = Q_{ac}U_0K(r)^* +
r\int_\dR dE^{ac}_0(\zeta)V\frac{\overline{\zeta}}{I - r\overline{\zeta}U} Q_{ac}U_0K(r)^*
\eed
which yields 
\bed
\gO_-(r)^*Q_{ac}U_0K(r)^* = Q_{ac}U_0K(r)^* +
rU^*_0\int_\dR dE^{ac}_0(\zeta)V\frac{1}{I - r\overline{\zeta}U} Q_{ac}U_0K(r)^*.
\eed
Using the notation \eqref{2.100} we obtain
\be\label{2.119}
\gO_-(r)^*Q_{ac}U_0K(r)^* = Q_{ac}U_0K(r)^* +
rU^*_0\Xi(r)^*Q_{ac}U_0K(r)^*.
\ee
Obviously we have
\be\label{2.120}
(\Phi Q_{ac}U_0K(r)^*\Phi^{-1} \wh f)(\zeta) = Q_{ac}(\zeta)\zeta\int_\dT
K(r;\xi,\zeta)^*\wh f(\xi)d\nu(\xi),
\ee
$\wh f \in L^2(\dT,d\nu(\zeta),\sh(\zeta))$. Using \eqref{2.107} we find
\bead
U^*_0\Xi(r)^*Q_{ac}U_0K(r)^* =
rU^*_0\int_\dT dE^{ac}_0(\zeta)CZ(r\zeta)^*\int_\dT C dE^{ac}_0(\xi)(1- r\overline{\zeta}\xi)^{-1}Q_{ac}U_0K(r)^*
\eead
which yields
\begin{align*}
&(\Phi U^*_0\Xi(r)^*Q_{ac}U_0K(r)^*\Phi^{-1}\wh f)(\zeta) \\
& =
r\left(\Phi U^*_0\int_\dT dE^{ac}_0(\zeta)CZ(r\zeta)^*\Phi^{-1}
\Phi\int_\dT C dE^{ac}_0(\xi)(1- r\overline{\zeta}\xi)^{-1}\Phi^{-1} 
\Phi Q_{ac}U_0K(r)^*\Phi^{-1}\wh f\right)(\zeta),
\end{align*}
 Hence
\bead
\lefteqn{
(\Phi U^*_0\Xi(r)^*Q_{ac}U_0K(r)^*\Phi^{-1}\wh f)(\zeta) }\\
& =&r\overline{\zeta}\sqrt{Y(\zeta)}Z(r\zeta)^*
\int_\dT d\nu(\xi) \sqrt{Y(\xi)}(1- r\overline{\zeta}\xi)^{-1}
Q_{ac}(\xi)\xi\int_\dT d\nu(\zeta')K(r;\zeta',\xi)^*\wh f(\zeta').
\nonumber
\eead
Since 
$K(r;\zeta,\xi) := \sqrt{Y(\zeta)}Z(r\zeta)^*\sqrt{Y(\xi)}$ by
definition we get
\bead
\lefteqn{
(\Phi U^*_0\Xi(r)^*Q_{ac}U_0K(r)^*\Phi^{-1}\wh f)(\zeta) = }\\
& &
r\int_\dT d\nu(\xi) \frac{\overline{\zeta}\xi K(r;\zeta,\xi)}{1- r\overline{\zeta}\xi}
Q_{ac}(\xi)\int_\dT d\nu(\zeta')K(r;\zeta',\xi)^*\wh f(\zeta'),\nonumber
\eead
 Finally, by definition \eqref{2.112} we find
\bea\label{2.125}
(\Phi U^*_0\Xi(r)^*Q_{ac}U_0K(r)^*\Phi^{-1}\wh f)(\zeta) = 
\int_\dT d\nu(\xi)\int_\dT d\nu(\zeta')\frac{\overline{\zeta}\xi M(r;\zeta,\xi,\zeta')}{1- r\overline{\zeta}\xi}\wh f(\zeta').
\eea
From \eqref{2.96} and \eqref{2.119} it follows
\bed
J_2(r,\varepsilon) = -\tr(\rho^\varepsilon_{ac} Q_{ac}U_0K(r)^*) -
r\tr(\rho^\varepsilon_{ac} U^*_0\Xi(r)^*Q_{ac}U_0K(r)^*).
\eed
Taking into account \eqref{2.120} and \eqref{2.125} we obtain
\bea\label{2.127}
J_2(r,\varepsilon) & = & -\int_\dT d\nu(\zeta)\zeta\tr(\rho^\varepsilon_{ac}(\zeta)Q_{ac}(\zeta)K(r;\zeta,\zeta)^*)\\
&- & 
r\int_\dT\int_\dT d\nu(\zeta)d\nu(\xi)\frac{\overline{\zeta}\xi}{1-
  r\overline{\zeta}\xi}\tr(\rho^\varepsilon_{ac}(\zeta) M(r;\zeta,\xi,\zeta)).
\nonumber 
\eea
From \eqref{2.91}, \eqref{2.115} and \eqref{2.127} we get
\bead
\lefteqn{
-2J(r,\varepsilon)=
\int_\dT d\nu(\zeta)\overline{\zeta}\tr\left(\rho^\varepsilon_{ac}(\zeta)K(r;\zeta,\zeta)Q_{ac}(\zeta)\right)
 }\\
& +&
\int_\dT d\nu(\zeta)\zeta\tr\left(\rho^\varepsilon_{ac}(\zeta)Q_{ac}(\zeta)K(r;\zeta,\zeta)^*\right) \\
& +& 
r \int_\dT d\nu(\zeta)\int_\dT d\nu(\xi)
\left\{\frac{1}{I - r\zeta\overline{\xi}} + \frac{\overline{\zeta}\xi}{1- r\overline{\zeta}\xi}\right\}
\tr\left(\rho^\varepsilon_{ac}(\zeta)M(r;\zeta,\xi,\zeta)\right)
\eead
which yields
\bead
\lefteqn{
-2J(r,\varepsilon) = \int_\dT d\nu(\zeta)\overline{\zeta}\tr\left(\rho^\varepsilon_{ac}(\zeta)K(r;\zeta,\zeta)Q_{ac}(\zeta)\right) }\\
& +&
\int_\dT d\nu(\zeta)\zeta\tr\left(\rho^\varepsilon_{ac}(\zeta)Q_{ac}(\zeta)K(r;\zeta,\zeta)^*\right)
\\
&+ &
2\pi\frac{r}{1+r} \frac{1-r^2}{2\pi}\int_\dT d\nu(\zeta)\int_\dT d\nu(\xi)
\frac{1 + \overline{\zeta}\xi}{|I - r\zeta\overline{\xi}|^2}
\tr\left(\rho^\varepsilon_{ac}(\zeta)M(r;\zeta,\xi,\zeta)\right).
\eead
By \eqref{2.111} we get
\bead
\lefteqn{
-2J(r,\varepsilon) =
\int_{\dT\setminus\gD_*(\varepsilon)}
d\nu(\zeta)\overline{\zeta}\tr\left(\rho_{ac}(\zeta)K(r;\zeta,\zeta)Q_{ac}(\zeta)\right)}\\
& +& 
\int_{\dT\setminus\gD_*(\varepsilon)}
d\nu(\zeta)\zeta\tr\left(\rho_{ac}(\zeta)Q_{ac}(\zeta)K(r;\zeta,\zeta)^*\right)\\
&+ &
2\pi\frac{r}{1+r} \frac{1-r^2}{2\pi}
\int_{\dT\setminus\gD_*(\varepsilon)} 
d\nu(\zeta)\int_\dT d\nu(\xi)\frac{1 + \overline{\zeta}\xi}{|I - r\zeta\overline{\xi}|^2}
\tr\left(\rho_{ac}(\zeta)M(r;\zeta,\xi,\zeta)\right).
\eead
Using the representation $K(r;\zeta,\zeta) =
X_*(r;\zeta)\sqrt{Y(\zeta)}$ and taking into account \eqref{2.55a} we
find that
\bead
\lim_{r\uparrow 1}\int_{\dT\setminus\gD_*(\varepsilon)}
d\nu(\zeta)\overline{\zeta}\tr\left(\rho_{ac}(\zeta)K(r;\zeta,\zeta)Q_{ac}(\zeta)\right)
= 
\int_{\dT\setminus\gD_*(\varepsilon)}
d\nu(\zeta)\overline{\zeta}\tr\left(\rho_{ac}(\zeta)K(\zeta,\zeta)Q_{ac}(\zeta)\right)
\eead 
and
\bead
\lim_{r\uparrow 1}
\int_\dT
d\nu(\zeta)\zeta\tr\left(\rho_{ac}(\zeta)Q_{ac}(\zeta)K(r;\zeta,\zeta)^*\right)
= 
\int_{\dT\setminus\gD_*(\varepsilon)}
d\nu(\zeta)\zeta\tr\left(\rho_{ac}(\zeta)Q_{ac}(\zeta)K(\zeta,\zeta)^*\right).
\eead
Furthermore, using \eqref{2.112} we find that
\bead
\lefteqn{\hspace{-2,0cm}
\frac{1-r^2}{2\pi}\int_{\dT\setminus\gD_*(\varepsilon)} d\nu(\zeta)
\int_\dT d\nu(\xi)\frac{1 + \overline{\zeta}\xi}{|I -
 r\zeta\overline{\xi}|^2} \tr\left(\rho_{ac}(\zeta)M(r;\zeta,\xi,\zeta)\right)}\\
& =&
\int_\dT d\nu(\xi)
\frac{1-r^2}{2\pi}\int_\dT d\nu(\zeta)\frac{1 + \overline{\zeta}\xi}{|I - r\zeta\overline{\xi}|^2} F(r;\zeta,\xi)
\chi_{\dT\setminus\gD(\varepsilon)}(\zeta)
\eead
where
\bed
F(r;\zeta,\xi) := 
\tr\left(\rho_{ac}(\zeta)X_*(r;\zeta)\sqrt{Y(\xi)}Q_{ac}(\xi)\sqrt{Y(\xi)}X_*(r;\overline{\zeta})^*\right)
\eed
$\zeta \in \dT\setminus\gD_*(\varepsilon)$, $\xi \in \dT$ and $0 \le r < 1$. By \eqref{2.55a} we
get the estimate
\bed
|F(r;\zeta,\xi)| \le C_{X_*}(\varepsilon)^2\|\rho_{ac}\|\|Q_{ac}\|\tr(Y(\xi)), \;
\zeta \in \dT \setminus \gD_*(\varepsilon)), \;  0 \le r < 1, \;
\xi \in \dT. 
\eed
Hence 
\bed
\begin{split}
&\left|\frac{1-r^2}{2\pi}
\int_{\dT\setminus\gD_*(\varepsilon)} d\nu(\zeta)\frac{1 +\overline{\zeta}\xi}{|I - r\zeta\overline{\xi}|^2}F(r;\zeta,\xi)\right| 
\le \\
&
2\frac{1-r^2}{2\pi}\int_{\dT\setminus\gD_*(\varepsilon)} d\nu(\zeta)\frac{1}{|I -
  r\zeta\overline{\xi}|^2}|F(r;\zeta,\xi)| \le 2C_{X_*}(\varepsilon)^2\|\rho_{ac}\|\|Q_{ac}\|\tr(Y(\xi))
\end{split}
\eed
where $\tr(Y(\xi)) \in L^1(\dT,d\nu(\xi))$. Applying the Lebesgue
dominated convergence theorem we obtain
\bead
\lefteqn{\hspace{-2,0cm}
\frac{1-r^2}{2\pi}\int_{\dT\setminus\gD_*(\varepsilon)} d\nu(\zeta)
\int_\dT d\nu(\xi)\frac{1 + \overline{\zeta}\xi}{|I -
 r\zeta\overline{\xi}|^2} \tr\left(\rho_{ac}(\zeta)M(r;\zeta,\xi,\zeta)\right)}\\
& =&
\int_\dT d\nu(\xi)F(\xi,\xi)\chi_{\dT\setminus\gD(\varepsilon)}(\xi) =
\int_{\dT \setminus \gD_*(\varepsilon)}d\nu(\xi) F(\xi,\xi)
\eead
where 
\bed
F(\xi,\xi) :=
\tr\left(\rho_{ac}(\xi)X_*(\xi)\sqrt{Y(\xi)}Q_{ac}(\xi)\sqrt{Y(\xi)}X_*(\overline{\xi})^*\right)
= \tr(\rho_{ac}(\xi)M(\xi,\xi,\xi))
\eed
and  $M(\zeta,\zeta,\zeta) = \sL_1-\lim_{r\uparrow 1}M(r;\zeta,\zeta,\zeta)$ for a.e. $\xi \in \dT$. Summing up we obtain 
\bead
-2J(\varepsilon) := 2\lim_{r\uparrow 1} J(r,\varepsilon) & = &
\;\int_{\dT\setminus\gD_*(\varepsilon)}d\nu(\zeta)\overline{\zeta}\tr\left(\rho_{ac}(\zeta)K(\zeta,\zeta)Q_{ac}(\zeta)\right)\\
& +&
\int_{\dT\setminus\gD_*(\varepsilon)}d\nu(\zeta)\zeta\tr\left(\rho_{ac}(\zeta)Q_{ac}(\zeta)K(\zeta,\zeta)^*\right)\\
&+ & 2\pi\int_{\dT\setminus\gD_*(\varepsilon)} d\nu(\zeta)\tr\left(\rho_{ac}(\zeta)M(r;\zeta,\zeta,\zeta)\right).
\eead
By Corollary \ref{II.8} we verify that
\begin{align*}
-2J(\varepsilon) &=  i\int_{\dT\setminus\gD_*(\varepsilon)}
d\nu(\zeta)\tr\left(\rho_{ac}(\zeta)T(\zeta)^*Q_{ac}(\zeta)\right) \\
&- i\int_{\dT\setminus\gD_*(\varepsilon)} d\nu(\zeta)\tr\left(\rho_{ac}(\zeta)Q_{ac}(\zeta)T(\zeta)\right)
+2\pi\int_{\dT\setminus\gD_*(\varepsilon)} d\nu(\zeta) \tr\left(\rho_{ac}(\zeta)M(\zeta,\zeta,\zeta)\right).
\nonumber
\end{align*}
Since
$M(\zeta,\zeta,\zeta) = K(\zeta,\zeta)Q_{ac}(\zeta)K(\zeta,\zeta)^*$
one gets
$M(\zeta,\zeta,\zeta) = T(\zeta)^*Q_{ac}T(\zeta)$. Therefore
\bed
2J(\varepsilon) = \int_{\dT\setminus\gD_*(\varepsilon)}d\nu(\zeta)\tr\left(\rho_{ac}(\zeta)\gS(\zeta)\right)
\eed
where
\be\label{2.133}
\gS(\zeta) := -iT(\zeta)^*Q_{ac}(\zeta) + iQ_{ac}(\zeta)T(\zeta)- 2\pi
T(\zeta)^*Q_{ac}(\zeta)T(\zeta),
\quad \zeta \in \dT.
\ee
Using \eqref{2.81a} we obtain $\|\gS\|_{\sL_1} \in
L^1(\dT,d\nu(\zeta))$. Moreover, from \eqref{2.62} we get
\be\label{2.134}
T(\zeta) = \frac{I_{\sh(\gl)} -S(\zeta)}{2\pi i} 
\quad \mbox{and} \quad 
T(\zeta)^* = - \frac{I_{\sh(\gl)} -S(\zeta)^*}{2\pi i}. 
\ee
Inserting \eqref{2.134} into \eqref{2.133} we find
\bead
\lefteqn{
\gS(\zeta) := \frac{I_{\sh(\zeta)} - S(\zeta)^*}{2\pi} Q_{ac}(\zeta) + Q_{ac}(\zeta)\frac{I_{\sh(\zeta)}- S(\zeta)}{2\pi} +}\\
& &
2\pi \frac{I_{\sh(\gl)} -S(\zeta)^*}{2\pi i} Q_{ac}(\zeta)\frac{I_{\sh(\gl)} -S(\zeta)}{2\pi i} 
\nonumber
\eead
which yields 
\bed
\gS(\zeta) = \frac{1}{2\pi}\left\{Q_{ac}(\zeta) - S(\zeta)^*Q_{ac}(\zeta)S(\zeta)\right\}.
\eed
which proves
\bed
J(\varepsilon) =
\frac{1}{4\pi}\int_{\dT\setminus\gD_*(\varepsilon)} \tr\left(\rho_{ac}(\zeta)(Q_{ac}(\zeta) - S(\zeta)^*Q_{ac}(\zeta)S(\zeta)\right)d\nu(\zeta)
\eed
Using  $\|\gS(\zeta)\|_{\sL_1} \in L^1(\dT,d\nu(\zeta))$ and
\eqref{2.114} we immediately prove \eqref{2.90}.
\end{proof}
\bc\la{cor:EquilibriumFluxZero}
If the assumptions of Theorem \ref{II.9} are satisfied, then
\be\la{2.134a}
J = \frac{1}{4\pi}\int_\dT\tr\left(\left(\rho_{ac}(\zeta) -
    S(\zeta)\rho_{ac}(\zeta)S(\zeta)^*\right)Q_{ac}(\zeta)\right). 
\ee
Further, let $\phi: \dT \longrightarrow [0,\infty)$ be Borel measurable and bounded. If $\rho = \phi(U_0)$, then $J = 0$.
\ec
\begin{proof}
Using the fact that $S(\zeta) - I_{\sh(\zeta)} \in \sL_1(\sh(\zeta))$ for
a.e $\zeta \in \dT$ with respect to $\nu$ one immediately shows that
\eqref{2.134a} follows from \eqref{2.90}.

If $\rho = \phi(U_0)$, then $\rho_{ac} = \phi(U^{ac}_0)$ which yields $\rho_{ac}(\zeta) =
\phi(\zeta)I_{\sh(\zeta)}$ for a.e. $\zeta \in \dT$ with respect to
$\nu$. Inserting $\rho_{ac}(\zeta) = \phi(\zeta)I_{\sh(\zeta)}$ into
\eqref{2.134a} we prove $J=0$. 
\end{proof}

\subsection{Self-adjoint operators}\la{III.2}

Let $H_0$ and $H$ be self-adjoint operators on the separable Hilbert
space $\sH$. If the condition
\be\label{2.137}
(H + i)^{-1} - (H_0 + i)^{-1}\in  \sL_1(\sH)
\ee
is satisfied, then the pair $\cS' = \{H,H_0\}$ is called a $\sL_1$-scattering system.
If $\cS' = \{H,H_0\}$ is a $\sL_1$-scattering system, then the wave operators 
\bed
W_\pm := \slim_{t\to\pm\infty}e^{itH}e^{-itH_0}P^{ac}(H_0)
\eed
exist and are complete. The scattering operator is defined by $S' :=
W^*_+W_-$.

A bounded self-adjoint operator $Q$ commuting with $H_0$ is called a
charge for $\cS'$. A non-negative bounded operator $\rho$ commuting with $H_0$
is called a density operator for $\cS'$. To define  the current $J'$ for $\cS'$
we assume that $(I + H^2_0)\rho$ is a bounded operator. Under this
assumption the current $J'$ is defined by 
\be\la{2.144}
J' := -i\tr\left(W_-(I+ H^2_0)\rho W^*_-(H-i)^{-1}[H,Q](H+i)^{-1}\right).
\ee
Using \eqref{2.5x} we have that $(H-i)^{-1}[H,Q'](H+i)^{-1} \in \sL_1(\sH)$ which shows that
the current is well defined. The definition \eqref{2.144} is in
accordance with \cite{Pillet2007}. Indeed, from definition \eqref{2.144} we
formally get $J' = -i\tr(W_-\rho W^*_-[H,Q])$. 
\bt	\label{thm:LandBuett}
Let $\cS' = \{H,H_0\}$ be a $\sL_1$-scattering system. 
Further, let $Q$ be a charge and let $\rho$ be a density operator for
$\cS'$ such that $(I + H^2_0)\rho$ is a bounded operator. 
Further, let $\Pi(H^{ac}_0)$ a spectral representation of $H^{ac}_0$
such that $Q_{ac}$ and $\rho_{ac}$ are represented by multiplication
operators $M_{Q'_{ac}}$ and $M_{\rho'_{ac}}$ induced by the 
measurable families $\{Q'_{ac}(\gl)\}_{\gl \in \dR}$ and
$\{\rho'_{ac}(\gl)\}_{\gl \in \dR}$, respectively. If $\sigma_{sc}(H) = \emptyset$,
then
\be\label{eq:LandBuett}
J'  = \frac{1}{2\pi}\int_\dR \tr\left(\rho'_{ac}(\gl)(Q'_{ac}(\gl)-S'(\gl)^*Q'_{ac}(\gl)S'(\gl))\right)d\gl
\ee
where $\{S'(\gl)\}_{\gl\in\dR}$ is the scattering matrix with respect
to the spectral representation $\Pi(H^{ac}_0)$.
\et
\begin{proof}
Let us introduce the Cayley transforms
\bed
U := (i - H)(i+ H)^{-1} 
\quad \mbox{and} \quad
U_0 := (i - H_0)(i + H_0)^{-1}.
\eed
The pair $\cS = \{U,U_0\}$ is a $\sL_1$-scattering system if and only if $\cS'$ is $\sL_1$-scattering system.
By the invariance principle for wave operators one verifies that $W_\pm
= \gO_\pm$ which yields $S = S'$. Obviously, $Q$ is a charge for
$\cS$  and $\rho$ is a density operator for $\cS$. 
A straightforward computation (compare with \eqref{mitlef1}) shows that
\be\label{2.142}
 J '= -\frac{1}{2}\tr(\gO_-\rho U^*_0\gO^*_-[V,Q]) =
-i\tr\left(W_-\rho W^*_-(H-i)^{-1}[H,Q](H+i)^{-1}\right).
\ee
Let $\Pi(U^{ac}_0)$ be the spectral representation of Appendix
\ref{B}. Assume that the operators $Q_{ac}$, $\rho_{ac}$ and $S =
\gO^*_+\gO_-$ are represented in $\Pi(U^{ac}_0)$ by the multiplication
operators $M_{Q_{ac}}$, $M_{\rho_{ac}}$ and $M_S$
induced by the measurable families $\{Q_{ac}(\zeta)\}_{\zeta\in\dT}$,
$\{\rho_{ac}(\zeta)\}_{\zeta \in \dT}$ and
$\{S(\zeta)\}_{\zeta\in\dT}$, respectively.

Using the spectral representation $\Pi(H^{ac}_0) = \{L^2(\dR,d\gl,\sh'(\gl)),M,\Phi'\}$ 
of Appendix \ref{C} one gets that $Q_{ac}$, $\rho_{ac}$ and $S$ are
presented in $\Pi(H^{ac}_0)$ by multiplication
operators $M_{Q'_{ac}}$, $M_{\rho'_{ac}}$ and $M_{S'}$
induced by the measurable families $\{Q'_{ac}(\gl)\}_{\gl\in\dR}$,
$\{\rho'_{ac}(\gl)\}_{\gl\in \dR}$ and $\{S'(\gl)\}_{\gl\in\dR}$,
respectively. Notice that both families are related by 
\bead
Q'_{ac}(\gl) & = & Q_{ac}(e^{2i\arctan(\gl)}), \quad \gl \in \dR,\\
\rho'_{ac}(\gl) & = & \rho_{ac}(e^{2i\arctan(\gl)}), \quad \gl \in \dR,\\
S'(\gl) & = & S(e^{2i\arctan(\gl)}),\quad \gl \in \dR.
\eead
Taking into account Theorem \ref{II.9} we get
\bea\la{2.143}
-\frac{1}{2}\tr(\gO_-\rho U^*_0\gO^*_-[V,Q]))
= 
\frac{1}{2\pi}\int_\dR \tr\left(\rho'_{ac}(\gl)(Q'_{ac}(\gl)-S'(\gl)^*Q'_{ac}(\gl)S'(\gl))\right)\frac{d\gl}{1+\gl^2}.
\eea
Finally, replacing $\rho$ by $(I + H^2_0)\rho$ we obtain
\eqref{eq:LandBuett} from \eqref{2.143}.
\end{proof}
\bc\la{cor:LandBuett1}
If the assumptions of Theorem \ref{thm:LandBuett} are satisfied, then
\bed
J' = \frac{1}{2\pi}\int_\dR \tr\left(\left(\rho'_{ac}(\gl) -
S'(\gl)\rho'_{ac}(\gl)S'(\gl)^*\right)Q'_{ac}(\gl)\right)d\gl.
\eed
Further, let $\phi: \dR \longrightarrow [0,\infty)$ be Borel measurable and
bounded. If $(1+\gl^2)\phi(\gl)$, $\gl \in \dR$, is bounded and $\rho' = \phi(H_0)$ , then $J' = 0$.
\ec
The proof of Corollary \ref{cor:LandBuett1} follows from Corollary
\ref{cor:EquilibriumFluxZero}.

The charge $Q$ was defined as a bounded self-adjoint
operator. However, this definition is usually not sufficient in
applications, cf. below. In \cite[Definition 3.3]{Pillet2007}
the notion of tempered charge charge was introduce. An unbounded
self-adjoint operator $Q$ is called a tempered charge if $Q$
commutes with $H_0$ and for any bounded Borel set $\gL$ of $\dR$ the
truncated charge $Q_\gL := QE_0(\gL)$ is bounded where $E_0(\cdot)$ is the spectral measure of $H_0$.
For tempered charges we set
\bed
J'_\gL := -i\tr(W_-(I + H^2_0)\rho W^*_-(H-i)^{-1}[H,Q_\gL](H+i)^{-1}), \quad Q_\gL
:= QE_0(\gL).
\eed
Since $[Q,H_0]=0$, we can decompose 
$Q' = Q_{ac} \oplus Q_s$. Let $\Pi(H^{ac}_0) =
\{L^2(\dT,d\gl,\sh'(\gl)),M',\Phi'\}$
be a spectral representation of $H^{ac}_0$. Then there is a
measurable family $\{Q'_{ac}(\gl)\}_{\gl\in\dR}$ of bounded operators
such that $Q_{ac}$ is unitarily equivalent to the multiplication
operator $M_{Q'_{ac}}$ where
\bead
(M_{Q'_{ac}}\wh{f'})(\gl) & := & Q'_{ac}(\gl)\wh{f'}(\gl), \quad \wh{f'} \in \dom(M_{Q'_{ac}}), \quad \gl \in \dR,\\
\dom(M_{Q'_{ac}}) & := & \{\wh{f'} \in
L^2(\dR,d\gl,\sh'(\gl):  Q'_{ac}(\gl)\wh{f'}(\gl) \in
L^2(\dR,d\gl,\sh'(\gl)\}.
\eead
Obviously, one gets $Q'_{\gL,ac}(\gl) = Q'_{ac}(\gl)\chi_\gL(\gl)$, $\gl \in
\dR$. If $Q$ is a tempered charge, then $Q_{ac}$ is a tempered charge
for $H^{ac}_0$, that is $\|Q_{ac}E^{ac}_0(\gL)\|_\sH <
\infty$. Therefore, for a tempered charge one has
\be\la{2.53a}
\sup_{\gL \in \cB_b(\dR)}\ess-sup_{\gl\in\gL}
  \|Q'_{ac}(\gl)\|_{\sh'(\gl)} < \infty
\ee
where $\ess-sup$ means the essential spectrum with respect to the
Lebesgue measure on $\dR$. In the following we denote the set of all
bounded Borel sets of $\dR$ by  $\cB_b(\dR)$.
\bc
Let $\cS' = \{H,H_0\}$ be a $\sL_1$-scattering system. 
Further, let $Q$ be a tempered charge and let $\rho$ be a density
operator. If
\be\la{2.53b}
\sup_{\gL \in \cB_b(\dR)} \|Q E_0(\gL)\|_\sH\|(I + H^2_0)\rho E_0(\gL)\|_\sH < \infty
\ee
then the limit $J' := \lim_{L\to\infty}J'_{(-L,L)}$
exists and the formula \eqref{eq:LandBuett} is valid.
\ec
\begin{proof}
Applying Theorem \ref{thm:LandBuett} we find
\bed
J'_\gL = \frac{1}{2\pi}\int_\gL\tr(\rho'_{ac}(\gl)(Q'_{ac}(\gl) -
S'(\gl)^*Q'_{ac}(\gl)S'(\gl)))d\gl, \quad \gL \in \cB_b(\dR).
\eed
From \eqref{2.53b} which yields
\bed
\sup_{\gL \in \cB_b(\dR)} \|Q_{ac} E_0(\gL)\|_\sH\|(I + H^2_0)\rho_{ac} E_0(\gL)\|_\sH < \infty
\eed
which yields
\bead
\lefteqn{
\|Q_{ac} E_0(\gL)\|_\sH\|(I + H^2_0)\rho_{ac} E_0(\gL)\|_\sH =}\\
& &
\ess-sup_{\gl\in\gL}\|Q'_{ac}(\gl)\|_{\sh'(\gl)}
\ess-sup_{\gl\in\gL}(1+\gl^2)\|\rho'_{ac}(\gl)\|_{\sh'(\gl)}.
\eead
Hence
\bead
\lefteqn{
\sup_{\gL \in \cB_b(\dR)}\|Q_{ac} E_0(\gL)\|_\sH\|(I + H^2_0)\rho_{ac} E_0(\gL)\|_\sH =}\\
& &
\sup_{\gL \in \cB_b(\dR)}\left\{
\ess-sup_{\gl\in\gL}\|Q'_{ac}(\gl)\|_{\sh'(\gl)}
\ess-sup_{\gl\in\gL}(1+\gl^2)\|\rho'_{ac}(\gl)\|_{\sh'(\gl)}
\right\}.
\eead
This gives
\bed
\sup_{\gL \in
  \cB_b(\dR)}\ess-sup_{\gl\in\gL}\left\{(1+\gl^2)\|Q'_{ac}(\gl)\|_{\sh'(\gl)}\|\rho'_{ac}(\gl)\|_{\sh'(\gl)}\right\} < \infty.
\eed
In particular, we have
\be\la{2.56r}
\sup_{L > 0}\;\ess-sup_{\gl\in (-L,L)}\left\{(1+\gl^2)\|Q'_{ac}(\gl)\|_{\sh'(\gl)}\|\rho'_{ac}(\gl)\|_{\sh'(\gl)}\right\} < \infty. 
\ee
Using the definition  $T'(\gl) := \frac{1}{2\pi}(I_{\sh'(\gl)} -
S'(\gl))$, $\gl \in \dR$, we find the relation $T'(\gl) = T(e^{2i\arctan(\gl)})$ for
a.e. $\gl \in \dR$. Taking into account \eqref{2.81a} we get the
estimate 
\be\la{2.57r}
\int_\dR \|T'(\gl)\|_{\sL_1} \frac{d\gl}{1+\gl^2} \le 2\|(H + i)^{-1}
- (H_0 + i)^{-1}\|_{\sL_1}.
\ee
Since 
\bead
\lefteqn{
Q'_{ac}(\gl) - S'(\gl)Q'_{ac}S'(\gl) = }\\
& &
2\pi
i\left\{T'(\gl)Q'_{ac}(\gl) + Q_{ac}(\gl)T'(\gl) -2\pi iT'(\gl)Q_{ac}(\gl)T'(\gl)\right\}
\eead
for a.e. $\gl \in \dR$ we find
\bead
\lefteqn{
\left\|\rho'_{ac}(\gl)(Q'_{ac}\left(\gl) -
    S'(\gl)Q'_{ac}S'(\gl)\right)\right\|_{\sL_1} \le}\\
& &
\left(2 + \tfrac{1}{\pi}\right)\|\rho'_{ac}(\gl)\|_{\sh'(\gl)}\|Q'_{ac}(\gl)\|_{\sh'(\gl)}\|T'(\gl)\|_{\sL_1}
\eead
for a.e. $\gl \in \dT$ where we have used that $\|T'(\gl)\|_{\sh'(\gl)} \le \tfrac{1}{\pi}$.
Using \eqref{2.56r} and \eqref{2.57r} we verify that the integral 
\bed
J'_\dR := \int_\dR\tr(\rho'_{ac}(\gl)(Q'_{ac}(\gl) -
S'(\gl)^*Q'_{ac}(\gl)S'(\gl)))d\gl
\eed
exists and is finite. Hence 
\bed
\lim_{L\to\infty}J'_{(-L,L)} = \lim_{L\to\infty}\frac{1}{2\pi}\int^L_{-L}\tr(\rho'_{ac}(\gl)(Q'_{ac}(\gl) -
S'(\gl)^*Q'_{ac}(\gl)S'(\gl)))d\gl  = J'_\dR
\eed
which completes the proof.
\end{proof}

\section{Examples}\la{IV}

Let us consider examples where the it is important that the
Hamiltonian is not semibounded from below.

\subsection{Landauer-B\"uttiker formula for dissipative operators}

We consider
the Schr\"odinger-type operator $K$  in the Hilbert space $\sK
= L^2((a,b))$ defined by
\bed
\dom(K) :=
\left\{
g \in W^{1,2}((a,b)): 
\begin{matrix}
\frac{1}{m(x)}g'(x) \in W^{1,2}((a,b))\\
\left(\frac{1}{2m}g'\right)(a) = -\gk_a g(a)\\
\left(\frac{1}{2m}g'\right)(b) = \gk_b g(b)
\end{matrix}
\right\}
\eed
and
\bed
(Kg)(x) = l(g)(x), \quad g \in \dom(K),
\eed
where
\bed
l(g)(x) = -\frac{d}{dx}\frac{1}{2m(x)}\frac{d}{dx}g(x) + V(x)g(x),
\quad x \in (a,b),
\eed
$V \in L^\infty((a,b))$ and $m(x) > 0$ is real function such that
$m\in L^\infty((a,b))$ and $\frac{1}{m} \in
L^\infty((a,b))$. Furthermore, we assume $\gk_a,\gk_b \in \dC_+ = \{z
\in\dC: \imag(z) > 0\}$. The operator $K$ is maximal dissipative and
completely non-self-adjoint. Its spectrum consists of non-real
isolated eigenvalues in $\dC_-$ which accumulate at infinity.

To analyze the operator $K$ it is useful to introduce the elementary 
solutions $v_a(x,z)$ and $v_b(x,z)$,
\begin{align}
&l(v_a(x,z)) - zv_a(x,z) = 0, \quad v_a(a,z) = 1, 
    \quad \frac{1}{2m(a)}v'_a(a,z) = -\gk_a,  \la{4.3}\\
&l(v_b(x,z)) - zv_b(x,z) = 0, \quad v_b(b,z) = 1,
    \quad \frac{1}{2m(b)}v'_b(b,z) = \gk_b, \la{4.4}
\end{align}
$x \in [a,b]$, $z \in \dC$, which always exist. 
The Wronskian of $v_a(x,z)$ and $v_b(x,z)$ is defined  by $W(z)$, i.e.
\be\la{2.6}
W(z) = v_a(x,z)\frac{1}{2m(x)}v'_b(x,z) - v_b(x,z)\frac{1}{2m(x)}v'_a(x,z).
\ee                                               
We note that the Wronskian does not depend on $x$. 
Obviously, the functions $v_{*a}(x,z)$ and $v_{*b}(x,z)$,
\be\la{2.7-1}
v_{*a}(x,z) := \overline{v_a(x,\overline{z})} \quad \mbox{and} \quad
v_{*b}(x,z) := \overline{v_b(x,\overline{z})}, 
\quad z \in \dC.
\ee
$x \in [a,b]$, $z \in \dC$, are also elementary solutions of 
\begin{align}
&l(v_{*a}(x,z)) - zv_{*a}(x,z) = 0, \quad v_{*a}(a,z) = 1, 
\quad \frac{1}{2m(a)}v'_{*a}(a,z) = -\overline{\gk_a}, \la{2.9} \\[2mm]
&l(v_{*b}(x,z)) - zv_{*b}(x,z) = 0, \quad v_{*b}(b,z) = 1, 
\quad \frac{1}{2m(b)}v'_{*b}(b,z) = \overline{\gk_b},\la{2.10}
\end{align}
$x \in [a,b]$. The Wronskian of $v_{*a}(x,z)$ and $v_{*b}(x,z)$ is
denoted by $W_*(z)$ and is also independent from $x$. Using the
elementary solutions one gets the representation
\bea\la{4.9}
\lefteqn{\hspace{-1.0cm}((H - z)^{-1}f)(x) = }\\
& & -\frac{v_b(x,z)}{W(z)}\int^x_a dy\,v_a(y,z)f(y) - 
\frac{v_a(x,z)}{W(z)}
\int^b_x dy\,v_b(y,z)f(y),\nonumber
\eea
for $z \in \varrho(H)$ and $f \in L^2([a,b])$ and
\bea\la{2.13}
\lefteqn{\hspace{-1.0cm}((H^* - z)^{-1}f)(x) = }\\
& & -\frac{v_{*b}(x,z)}{W_*(z)}\int^x_a dy\,v_{*a}(y,z)f(y) - 
\frac{v_{*a}(x,z)}{W_*(z)}
\int^b_x dy\,v_{*b}(y,z)f(y),\nonumber
\eea
for $z \in \varrho(H^*)$ and $f \in L^2([a,b])$, see \cite{KNR2}.

Since $H$ is completely non-self-adjoint the maximal dissipative
operator $H$ can be completely characterized by its characteristic
function $\gth_K(z)$, $z \in \varrho(H) \cap \gr(H^*)$. The definition of
the characteristic function relies on the so-called boundary operators $T(z):
\sK \longrightarrow \dC^2$, $z \in \gr(H)$ and $T_*(z): \sK
\longrightarrow \dC^2$, $z \in \gr(H^*)$, which are defined in
\cite{KNR2}. Introducing representations
\be\la{2.14}
\gk_a = q_a + \frac{i}{2}\ga_a^2 \quad \mbox{and} \quad 
\gk_b = q_b + \frac{i}{2}\ga_b^2, \quad \ga_a,\ga_b > 0,
\ee
the boundary operators are defined by
\be\la{4.12}
T(z)f := 
\left(\ba{c} 
\ga_b((H - z)^{-1}f)(b)\\
-\ga_a((H - z)^{-1})f(a) 
\ea\right)
\ee
and
\be\la{2.16}
T_*(z)f := \left(
\ba{c}
\ga_b((H^* - z)^{-1}f)(b)\\
-\ga_a((H^* - z)^{-1}f)(a)
\ea
\right),
\ee
$f \in L^2([a,b])$. Using the resolvent representations \eqref{4.9}
and \eqref{2.13} we obtain
\be\la{2.17}
T(z)f = \frac{1}{W(z)}\left(
\ba{c}
-\ga_b \int^b_a dy \; v_a(y,z)f(y)\\
\ga_a \int^b_a dy \; v_b(y,z)f(y)
\ea
\right)
\ee
and
\be\la{2.17-1}
T_*(z)f = \frac{1}{W_*(z)}\left(
\ba{c}
-\ga_b\int^b_a dy \; v_{*a}(y,z)f(y)\\
\ga_b\int^b_a dy \; v_{*b}(y,z)f(y)
\ea
\right),
\ee
$f \in L^2([a,b])$. The adjoint operators are given by

\bea\la{2.17-2}
\left(T(z)^*\xi\right)(x) & = & \frac{1}{\overline{W(z)}}
\left(-\ga_b \overline{v_a(x,z)},\ga_a\overline{v_b(x,z)}\right)\xi\\
          & = & \frac{1}{W_*(\overline{z})}
\left(-\ga_b v_{*a}(x,\overline{z}),\ga_av_{*b}(x,\overline{z})\right)\xi,\nonumber 
\eea
and
\bea\la{2.17-3}
\left(T_*(z)^*\xi\right)(x) & = & \frac{1}{\overline{W_*(z)}}
\left(-\ga_b\overline{v_{*a}(x,z)},\ga_a\overline{v_{*b}(x,z)}\right)\xi\\
            & = & \frac{1}{W(\overline{z})}
\left(-\ga_bv_a(x,\overline{z}),\ga_av_b(x,\overline{z})\right)\xi,\nonumber
\eea
where%
\be\la{2.17-4}
\xi = \left(
\ba{c}
\xi^b\\
\xi^a
\ea
\right) \in \dC^2.
\ee
The characteristic function
$\gT_K(\cdot)$ of the maximal dissipative
operator $H$ is a two-by-two matrix-valued function which
satisfies the relation
\be\la{2.18}
\gT_K(z)T(z)f = T_*(z)f, \quad z \in \gr(H) \cap \gr(H^*), \quad
\ga_a,\ga_b > 0,
\ee
$f \in L^2([a,b])$. It depends
meromorphically on $z \in \gr(H) \cap \gr(H^*)$ and is contractive in $\dC_-$, i.e.
\be\la{2.19}
\|\gT_K(z)\| \le 1 \quad \mbox{for} \quad z \in \dC_-.
\ee
Using the elementary solutions the characteristic
function $\gT_K(\cdot)$ takes the form
\be\la{2.24}
\gT_K(z) = I_{\dC^2} + i\frac{1}{W_*(z)}
\left(\ba{cc} 
\ga^2_bv_{*a}(b,z) & -\ga_b\ga_a \\
-\ga_b\ga_a         & \ga^2_av_{*b}(a,z) 
\ea\right).       
\ee                                                                
for $z \in \gr(H) \cap \gr(H^*)$, cf. \cite{KNR2}

Since the operator $K$ is maximal dissipative there is a larger
Hilbert space $\sH$ and a self-adjoint operator $H$ such that $\sK$ is
embed in $\sH$ and the relation
\bed
P^{\sH}_{\sK}(H- z)^{-1}\upharpoonright\sK = (K-z)^{-1}, \quad z \in
\dC_+,
\eed
is satisfied. The self-adjoint operator $H$ is called a self-adjoint
dilation of $K$. If the condition 
\bed
\clospa\{(H-z)^{-1}\sK: z \in \dC \setminus \dR\} = \sH
\eed
is satisfied, then $H$ is called a minimal self-adjoint dilation
$K$ of $H$. Minimal self-adjoint dilations of maximal dissipative operators are
determined up to a certain isomorphism, in particular, all minimal
self-adjoint dilations are unitarily equivalent. 

In the present case the minimal self-adjoint dilation of the maximal
dissipative operator $H$ can be constructed in an explicit manner. In 
accordance with \cite{KNR2} we introduce the larger Hilbert space
\be\la{2.28}
\sH = \sD_- \oplus \sK \oplus \sD_+
\ee
where $\sD_\pm := L^2(\dR_\pm,\dC^2)$. Introducing the graph $\gO$,
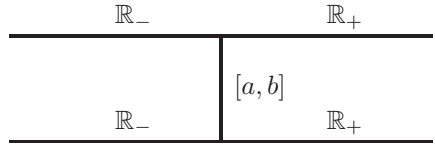
\begin{center}
\setlength{\unitlength}{1pt}
\begin{picture}(50,40)\thicklines
\put(-80,20){\line(1,0){80}} 
\put(-40,25){\makebox{$\dR_-$}}
\put(0,20){\line(1,0){80}} 
\put(40,25){\makebox{$\dR_+$}}
\put(-80,-20){\line(1,0){80}} 
\put(-40,-15){\makebox{$\dR_-$}}
\put(0,-20){\line(1,0){80}} 
\put(40,-15){\makebox{$\dR_+$}}
\put(0,-20){\line(0,1){40}} 
\put(5,-2){\makebox{$[a,b]$}}
\end{picture}
\end{center}

\vspace{30pt}
\noindent
one can write the Hilbert space $\sH$ as $L^2(\hat{\gO})$. Further, we define
\be\la{2.29}
\vec{g} := g_- \oplus g \oplus g_+
\ee
where
\be\la{2.30}
g_-(x) := \left(\ba{c}
g^b_-(x)\\
g^a_-(x)
\ea\right) \qquad \mbox{and} \quad
g_+(x) := \left(\ba{c}
g^b_+(x)\\
g^a_+(x)
\ea\right)
\ee
for $x \in \dR_-$ and $x \in \dR_+$, respectively. 
Let us introduce the matrices $K^a_\pm$ and $K^b_\pm$ which are defined by
\be\la{2.31}
K^a_- := 
\left(\ba{cc}
0   &   0 \\
1   &  \gk_a
\ea\right) 
\qquad \mbox{and} \qquad 
K^a_+ := 
\left(\ba{cc}
0  & 0 \\
1  & \overline{\gk_a}
\ea\right)
\ee
as well as
\be\la{2.32}
K^b_- := 
\left(\ba{cc}
1   &   -\gk_b \\
0   &  0
\ea\right) 
\qquad \mbox{and} \qquad 
K^b_+ := 
\left(\ba{cc}
1  & -\overline{\gk_b}\\
0  & 0
\ea\right).
\ee
Further we set
\bed
\gL :=
\begin{pmatrix}
\ga_b & 0\\
0 & \ga_b
\end{pmatrix}
\eed
Using these notations the self-adjoint dilation $K$ is defined by
\be\la{2.33}
\dom(H) := \left\{\vec{g} \in \sH: 
\ba{l}
g_\pm \in W^{1,2}(\dR_\pm,\dC^2),\\
g,\frac{1}{m}g' \in W^{1,2}([a,b]),\\  
K^a_-g_a + K^b_-g_b = \gL g_-(0),\\
K^a_+g_a + K^b_+g_b = \gL g_+(0)
\ea
\right\}
\ee
and
\be\la{2.34}
H\vec{g} := -i\frac{d}{dx}g_- \oplus l(g) \oplus -i\frac{d}{dx}g_+, 
\qquad \vec{g} \in \dom(H),
\ee
where,
\be\la{2.34-1}
g_a = \left(\ba{c}
\frac{1}{2 m(a)}g'(a)\\
g(a)
\ea\right) \quad \mbox{and} \quad 
g_b = \left(\ba{c}
\frac{1}{2 m(b)}g'(b)\\
g(b)
\ea\right), 
\ee
With respect to a graph picture the operator $H$ looks like

\vspace{25pt}
\begin{center}
\setlength{\unitlength}{1pt}
\hspace{40pt}
\begin{picture}(50,40)\thicklines
\put(-160,30){\line(1,0){100}} 
\put(-60,28){$)$}
\put(-160,48){\makebox{$\ga_b g^b_-(0) = \frac{1}{2m(b)}g'(b) - \gk_bg(b)$}}
\put(-140,15){\makebox{$-i\frac{d}{dx}g^b_-$}}
\put(50,30){\line(1,0){100}} 
\put(48,28){\makebox{(}}
\put(20,48){\makebox{$\frac{1}{2m(b)}g'(b) - \overline{\gk_b}g(b) = \ga_b g^b_+(0)$}}
\put(100,15){\makebox{$-i\frac{d}{dx}g^b_+$}}
\put(-160,-30){\line(1,0){100}} 
\put(-60,-32){\makebox{$)$}}
\put(-165,-52){\makebox{$\ga_a g^a_-(0) = \frac{1}{2m(a)}g'(a) + \gk_ag(a)$}}
\put(-140,-20){\makebox{$-i\frac{d}{dx}g^a_-$}}
\put(50,-30){\line(1,0){100}} 
\put(48,-32){\makebox{$($}}
\put(20,-52){\makebox{$\frac{1}{2m(a)}g'(a) + \overline{\gk_a}g(a) = \ga_ag^a_+(0)$}}
\put(100,-20){\makebox{$-i\frac{d}{dx}g^a_+$}}
\put(0,-30){\line(0,1){60.0}} 
\put(10,0){\makebox{$l(g)$}}
\end{picture}
\end{center}

\vspace{60pt}
\noindent
We define another self-adjoint operator $H_0$ by setting 
$\ga_b = \ga_a = 0$. In this case we get
\bed
\dom(H_0) := \left\{\vec{g} \in \sH: 
\ba{l}
g_\pm \in W^{1,2}(\dR_\pm,\dC^2),\\
g,\frac{1}{m}g' \in W^{1,2}((a,b)),\\  
K^a_-g_a + K^b_-g_b = 0,\\
K^a_+g_a + K^b_+g_b = 0,\\
g_-(0) = g_+(0)
\ea
\right\}
\eed
and
\bed
H_0\vec{g} := -i\frac{d}{dx}g_- \oplus l(g) \oplus -i\frac{d}{dx}g_+, 
\qquad \vec{g} \in \dom(H_0),
\eed
Setting $\sD = \sD_- \oplus \sD_+ = L^2(\dR,\dC^2)$ we obtain
\bed
\sH = \sD \oplus \sK
\eed
and
\bed
H_0 = T \oplus K_0
\eed
where $T$ is the momentum operator given by $\dom(T) := W^{1,2}(\dR,\dC^2)$
\bed
(Tf)(x) := -i\frac{d}{dx}f(x), \quad f \in \dom(T),
\eed
and $K_0$ is defined by 
\bed
\dom(K_0) := \left\{\vec{g} \in \sH:
\begin{matrix}
\frac{1}{m}g' \in W^{1,2}((a,b))\\
(\frac{1}{2m}g)(b) = q_bg(b)\\
(\frac{1}{2m}g)(a) = - q_ag(a)\\
\end{matrix}
\right\}
\eed
Since the operator $K_0$ is discrete one gets $H^{ac}_0 = T$ and
$\sH^{ac}(H_0) = L^2(\dR,\dC^2)$. One easily checks that the
resolvent difference is a trace class operator. This is due to the
fact that both operators $H$ and $H_0$ are self-adjoint extensions of
the symmetric operator $\wt H$,
\bed
\dom(\wt H) :=
\left\{\vec{g} \in \sH:
\begin{matrix} 
g_\pm \in W^{1,2}(\dR_\pm,\dC^2)\\
g,\frac{1}{m}g' \in  W^{1,2}((a,b))\\
g_a = g_b = 0\\
g_\pm(0) = 0
\end{matrix}
\right\},
\eed
which has finite deficiency indices. Hence $\cS = \{H,H_0\}$ is trace
class scattering system. In particular, the wave operators
$W_\pm(H,H_0)$ exist and are complete.

One easily checks that $\Pi(H^{ac}_0) = \{L^2(\dR,d\gl,\dC^2),M,\cF\}$
where $M$ is the multiplication operator induced by the
independent variable $\gl$ and $\cF$ denotes the Fourier transform
\bed
(\cF f)(\gl)  = \frac{1}{2\pi}\int_\dR e^{-i\gl x} f(x)dx, \qquad f
\in L^2(\dR,dx,\dC^2).
\eed
It is known that the scattering operator $S(H,H_0) =
W_+(H,H_0)^*W_-(H,H_0)$ is unitarily equivalent to the multiplication
operator $M_{\gT^*}$ induced by the measurable family
$\{\gT(\gl)^*\}_{\gl\in\dR}$ in $L^2(\dR,d\gl,\dC^2)$ where 
\bed
\gT(\gl) = \lim_{\eta\to +0}\gT(\gl -
i\eta) =
\begin{pmatrix}
1 & 0\\
0 & 1
\end{pmatrix} +
i\frac{1}{W_*(\gl)}
\begin{pmatrix}
\ga^2_bv_{*a}(b,\gl) & -\ga_b\ga_a\\
-\ga_b\ga_a & \ga^2_av_{*b}(a,\gl)
\end{pmatrix}
\eed
which exist  and is contractive for $\gl \in \dR$.
Setting
\bed
\gth_b(\gl) := W(\gl) - i\ga^2_bv_a(b,\gl) 
\quad \mbox{and} \quad 
\gth_a(\gl) := W(\gl) - i\ga^2_av_b(a,\gl),
\eed
$\gl \in \dR$, we find the representation
\bed
\gT(\gl) = \frac{1}{\overline{W(\gl)}}
\begin{pmatrix}
\overline{\gth_b(\gl)} & -i\ga_b\ga_a\\
-i\ga_b\ga_a & \overline{\gth_a(\gl)}
\end{pmatrix}
\eed
and
\be\la{4.42}
\gT(\gl)^* = \frac{1}{W(\gl)}
\begin{pmatrix}
\gth_b(\gl) & i\ga_b\ga_a\\
i\ga_b\ga_a & \gth_a(\gl)
\end{pmatrix}.
\ee
Since $\gT(\gl)^*\gT(\gl) = I_{\dC^2}$ for $\gl \in \dR$ we obtain
\be\la{4.43}
1 = |\gth_b(\gl)|^2 + \ga^2_b\ga^2_a = |\gth_a(\gl)|^2 + \ga^2_b\ga^2_a
\quad \mbox{and} \quad
\gth_a(\gl) = \overline{\gth_b(\gl)}
\ee
for $\gl \in \dR$.

Let $\rho$ be a steady state for $H_0$. Obviously, the steady
state is unitarily equivalent to the multiplication $M_\rho$ induced
by a measurable family $\{\rho(\gl)\}_{\gl\in\dR}$ of non-negative
bounded self-adjoint operators acting in $\dC^2$. We use the representation
\be\la{4.44}
\rho(\gl) = 
\begin{pmatrix}
\rho_b (\gl) & \overline{\tau(\gl)}\\
\tau(\gl) & \rho_a(\gl)
\end{pmatrix} \ge 0, \quad \gl \in \dR.
\ee
Notice that $\rho(\gl) \ge 0$ if and only if the conditions
$\rho_b(\gl) \ge 0$, $\rho_a(\gl) \ge 0$ and
\bed
|\tau(\gl)|^2 \le \rho_b(\gl)\rho_a(\gl)
\eed
is satisfied for a.e. $\gl \in \dR$. Moreover, $\rho$ and $(I + H^2_0)\rho$ are bounded operators 
if and only the conditions
\bed
\ess-sup_{\gl\in\dR}\left\{\rho_b(\gl) + \rho_a(\gl) + |\tau(\gl)|\right\} <
  \infty.
\eed
and 
\be\la{4.45}
\ess-sup_{\gl\in\dR}(1+\gl^2)\left\{\rho_b(\gl) + \rho_a(\gl) + |\tau(\gl)|\right\} <
  \infty.
\ee
are satisfied, respectively. 

In \cite{KNR3} the current related to the self-adjoint operator
$H$ was calculated in accordance with \cite{LF1989}. To this end
the generalized incoming eigenfunctions $\psi(x,\gl,a)$ and
$\psi(x,\gl,b)$, $x \in \gO$, $\gga \in\{a,b\}$, $\gl \in \dR$ of $H$ were computed 
and the current $j_\rho(x,\gl)$ was defined by 
\bed
\begin{split}
j_\rho(x,\gl) := &
\mu_b(\gl)\imag\left(\frac{1}{m(x)}\overline{\psi(x,\gl,b)}m(x)\psi'(x,\gl,b)\right)
+\\
&\mu_a(\gl)\imag\left(\frac{1}{m(x)}\overline{\psi(x,\gl,a)}m(x)\psi'(x,\gl,a)\right)  
\end{split}
\eed
for $x \in \gO$, $\gl \in \dR$, where $\mu_b(\gl)$ and $\mu_a(\gl)$
are the eigenvalues of $\rho(\gl)$.  It turns out
that $j_\rho(x,\gl)$ is independent from $x$, that is $j_\rho(\gl) :=
j_\rho(x,\gl)$, and admits the representation
\bed
j_\rho(\gl) = \tr(\rho(\gl)C(\gl)), \quad \gl \in \dR
\eed 
where
\bed
C(\gl) := -\frac{1}{2\pi i}\frac{\ga_b\ga_a}{\overline{W(\gl)}}E\gT(\gl)^*, \quad \gl \in \dR,
\eed
and
\bed
E := 
\begin{pmatrix}
0 & 1 \\
-1 & 0
\end{pmatrix},
\eed
cf. Proposition 4.1 of \cite{KNR3}. If $\tr(\rho(\gl)) \in L^1(\dR,d\gl)$, then the full current $j_\rho$ is given by
\bed
j_\rho = \int_\dR j_\rho(\gl) d\gl
\eed
cf. Proposition 4.1 of \cite{KNR3}. Using \eqref{4.42} and \eqref{4.44} we find
\be\la{4.48}
j_\rho = \frac{1}{2\pi}\int_\dR \frac{-\ga^2_b\ga^2_a(\rho_b(\gl) - \rho_a(\gl)) + i\ga_b\ga_b(\tau(\gl)\gth_a(\gl) - \overline{\tau(\gl)}\gth_b(\gl))}{|W(\gl)|^2}\;d\gl.
\ee

Let us calculate the current in accordance with Theorem \ref{thm:LandBuett}. 
To define charges we note that $\sD$ admits the decomposition 
\bed
\sD = 
\begin{matrix}
\sD_b\\
 \oplus\\
\sD_a
\end{matrix}.
\eed
By $Q_b$ and $Q_a$ we denote the projections form $\sD$ onto $\sD_b$
and $\sD_a$, respectively. The  operators $Q_b$ and $Q_a$ commute with
$H_0$ and can be regarded as charges. The charge matrices are given
\bed
Q_b(\gl) = 
\begin{pmatrix}
1 & 0\\
0 & 0
\end{pmatrix}
\quad \mbox{and} \quad
Q_a(\gl) = 
\begin{pmatrix}
0 & 0\\
0 & 1
\end{pmatrix}, \quad \gl \in \dR.
\eed
Applying Theorem \ref{thm:LandBuett} we find
\bed
J^\cS_{\rho,Q_a} = \frac{1}{2\pi}\int_\dR 
\tr\left(\rho(\gl)\left(Q_a(\gl) - \gT(\gl)Q_a(\gl)\gT(\gl)^*\right)\right) d\gl.
\eed
A straightforward computation shows that
\bed
Q_a(\gl) - \gT(\gl)Q_a(\gl)\gT(\gl)^* =
\frac{1}{|W(\gl)|^2}
\begin{pmatrix}
-\ga^2_b\ga^2_a & i\ga_b\ga_a\gth_a(\gl)\\
-i\ga_b\ga_a\overline{\gth_a(\gl)} & \ga^2_b\ga^2_a
\end{pmatrix}.
\eed
Taking into account \eqref{4.44} we obtain
\bed
\begin{split}
\tr&(\rho(\gl)(Q_a(\gl) - \gT(\gl)Q_a(\gl)\gT(\gl)^*)) =\\
& \frac{1}{|W(\gl)|^2}\left(-\ga^2_b\ga^2_a(\rho_b(\gl) - \rho_a(\gl)) + i\ga_a\ga_b(\tau(\gl)\gth_a(\gl) - \overline{\tau(\gl)}\,\overline{\gth_a(\gl)})\right) 
\end{split}
\eed
which yields
\bed
J^\cS_{\rho,Q_a} = \frac{1}{2\pi}
\int_\dR\frac{-\ga^2_b\ga^2_a(\rho_b(\gl) - \rho_a(\gl)) + i\ga_a\ga_b(\tau(\gl)\gth_a(\gl) - \overline{\tau(\gl)}\,\overline{\gth_a(\gl)})}{|W(\gl)|^2}\,d\gl.
\eed
Using \eqref{4.43} we immediately get from 
\eqref{4.48} that $J^\cS_{\rho,Q_a} = j_\rho$. Comparing with \cite{KNR3} the proof is much shorter. Moreover, from
Proposition 4.1 of \cite{KNR3} we get that
\bed
|J^\cS_{\rho,Q_a}| \le \frac{1}{2\pi}\int_\dR \tr(\rho(\gl))\,d\gl = \frac{1}{2\pi}\int_\dR(\rho_b(\gl) + \rho_a(\gl))\,d\gl
\eed
By \eqref{4.45} the last integral exists.

\subsection{Landauer-B\"uttiker formula for a
      pseudo-relativistic system}

We consider the Hilbert space $L^2(\dR,\dC^2)$ and the symmetric Dirac operator
\bed
(A\vec{f})(x) = 
\begin{pmatrix}
0 & -1\\
1 & 0
\end{pmatrix}
\frac{d}{dx}\vec{f}(x) + 
\begin{pmatrix}
a & 0 \\
0 & -a
\end{pmatrix}\vec{f}(x), \;\; \vec{f} \in \dom(A), \; x \in \dR,
\eed
where $a > 0$ and 
\bed
\dom(A) := \{\vec{f} \in W^{1,2}(\dR,\dC^2): \vec{f}(0) = 0\}
\eed
and
\bed
\vec{f} =
\begin{pmatrix}
f_1\\
f_2
\end{pmatrix}, \quad f_1,f_2 \in L^2(\dR,dx).
\eed
The deficiency indices $n_\pm(A)$ are equal two. The operator $A$ is completely non-self-adjoint.
The domain of the adjoint operator is given by 
\bed
\dom(A^*) = W^{1,2}(\dR_-,\dC^2) \oplus W^{1,2}(\dR_+,\dC^2).
\eed
Its Weyl function $M(z)$ was calculated in \cite{BMN2002}. One has
\bed
M(z) =
\begin{pmatrix}
i\frac{\sqrt{z + a}}{\sqrt{z-a}} & 0\\
0 & i\frac{\sqrt{z - a}}{\sqrt{z+a}}
\end{pmatrix}, \qquad z \in \dC_+,
\eed
where the cut  of the square root $\sqrt{\cdot}$ is fixed along the non-negative real axis.
We define a self-adjoint extension $H_0$ of $A$ by $H_0 = A^*\upharpoonright\dom(H_0)$,
\bed
\dom(H_0) = \{\vec{f}\in \dom(A^*): f_2(-0) = 0, \quad f_1(+0) = 0\}.
\eed
The operator $H_0$ is self-adjoint and absolutely continuous. Its
spectrum is given by
$\gs(H_0) = \gs_{ac}(H_0) = \dR \setminus (-a,a)$. It is not hard to
see that the $H_0$ has the form
\bed
H_0 = H_- \oplus H_+
\eed
where $H_\pm$ are are self-adjoint operators in $L^2(\dR_\pm,\dC^2)$,
respectively. A straightforward computation
shows that the operator $H_-$ and $H_+$ are unitarily equivalent
to the operator $K_-$,
\bea
(K_-f)(x) & := & i\frac{d}{dx}f(x) - a f(-x), \quad f \in \dom(K_-),\\
\dom(K_-) & := & \{W^{1,2}(\dR_-)\oplus W^{1,2}(\dR_+): f(-0) = -f(+0)\},
\eea
and $K_+$,
\bed
(K_+f)(x) :=  i\frac{d}{dx}f(x) - a f(-x), \quad f \in \dom(K_+) :=  W^{1,2}(\dR),
\eed
defined in $L^2(\dR)$, respectively.

The limit $M(\gl) := \lim_{y\to+0}M(\gl + iy)$ exist for every point
$\gl \in \dR \setminus \{-a,a\}$. One has
\bed
M(\gl) =
\begin{pmatrix}
i\frac{\sqrt{\gl + a}}{\sqrt{\gl-a}} & 0\\
0 & i\frac{\sqrt{\gl-a}}{\sqrt{\gl+a}}
\end{pmatrix}, \quad \gl \in \dR \setminus \{-a,a\}.
\eed
Hence
\bed
\imag(M(\gl)) =
\begin{pmatrix}
\sqrt{\frac{\gl + a}{\gl-a}} & 0\\
0 & \sqrt{\frac{\gl-a}{\gl+a}}
\end{pmatrix}, \quad \gl \in \dR \setminus [-a,a],
\eed
and $\imag(M(\gl)) = 0$ for $\gl \in (-a,a)$. We set
$\sh(\gl) := \ran(\imag(M(\gl)))$, $\gl \in \dR \setminus \{-a,a\}$.
Obviously, we get
\bed
\sh(\gl) =
\begin{cases} 
\dC^2 & \gl \in \dR \setminus [-a,a]\\
0 & \gl \in (-a,a).
\end{cases}
\eed
We consider the direct integral $L^2(\dR,d\gl,\sh(\gl))$. I turns out
that there is an isometry $\Phi$ acting from $\sH$ onto
$L^2(\dR,d\gl,\sh(\gl))$
such that the triplet $\Pi(H_0) = \{L^2(\dR,d\gl,\sh(\gl)),\cM,\Phi\}$ is a
spectral representation of $H_0$. 

Another self-adjoint extension $H$ of $A$ is defined
by choosing a self-adjoint operator $B$,
\bed
B =
\begin{pmatrix}
b_- & \overline{r}\\
r & b_+
\end{pmatrix}, \quad b_-,b_+ \in \dR,\quad r \in \dC,
\eed
acting on $\dC^2$ and setting
\bed
\dom(H) := \left\{\vec{f} \in \dom(A^*):
\ba{rcl}
f_1(-0) & = & b_-f_2(-0) + \overline{r}f_1(+0)\\
f_2(+0) & = & rf_2(-0) + b_+f_1(+0)
\ea
\right\}
\eed
The self-adjoint extension $H$ can be regarded as the Hamiltonian of
some point interaction at zero.
Since the deficiency indices of $A$ are finite
the resolvent difference of $H$ and $H_0$ is trace class operator. 

We consider the trace class scattering system $\cS = \{H,H_0\}$.
Following \cite{BMN2008} the scattering
matrix $\{S(\gl)\}_{\gl\in\dR}$ admits the representation
\bed
S(\gl) = I_{\sh(\gl)} + 2i\sqrt{\imag(M(\gl))}(B -
M(\gl))^{-1}\sqrt{\imag(M(\gl))}, 
\eed
$\gl \in \dR \setminus [-a,a]$. We find
\bed
(B - M(\gl))^{-1} = \frac{1}{\det(B-M(\gl))}
\begin{pmatrix}
b_+ - i\frac{\sqrt{\gl-a}}{\sqrt{\gl+a}} & -\overline{r}\\
-r & b_- - i\frac{\sqrt{\gl + a}}{\sqrt{\gl-a}}
\end{pmatrix}
\eed
for $\gl \in \dR \setminus [-a,a]$. The transition matrix
$\{T(\gl)\}_{\gl\in\dR}$ is defined $T(\gl) := S(\gl) - I_{\sh(\gl)}$,
$\gl \in \dR \setminus [-a,a]$, which yields
\bed
T(\gl) = 2i\sqrt{\imag(M(\gl))}(B -
M(\gl))^{-1}\sqrt{\imag(M(\gl))}, \quad \gl \in \dR \setminus [-a,a].
\eed
Using the representation
\bed
T(\gl) = 
\begin{pmatrix}
t_{--}(\gl) & t_{-+}(\gl)\\
t_{+-}(\gl) & t_{++}(\gl)
\end{pmatrix}
\eed
we find
\bead
t_{--}(\gl) & = & \frac{2i}{\det(B - M(\gl))}\left(b_+\frac{\sqrt{\gl+a}}{\sqrt{\gl-a}} -i\right)\\
t_{-+}(\gl) & = & -\overline{r}\frac{2i}{\det(B - M(\gl))}\\
t_{+-}(\gl) & = & -r\frac{2i}{\det(B - M(\gl))}\\
t_{++}(\gl) & = & \frac{2i}{\det(B - M(\gl))}\left(b_-\frac{\sqrt{\gl - a}}{\sqrt{\gl+a}} - i\right)
\eead
We set
\bed
\gs(\gl) := |t_{-+}(\gl)|^2 = |t_{+-}(\gl)|^2 = \frac{4|r|^2}{|\det(B
  - M(\gl))|^2}, \quad \gl \in \dR
\setminus [-a,a],
\eed
which is the cross section between the left- and right-hand scattering
channels. Since $\|T(\gl)\|_{\cB(\dC^2)} \le 2$, $\gl \in \dR
\setminus [-a,a]$, we find $\gs(\gl) \le 2$, $\gl \in \dR \setminus
[-a,a]$, which yields
\bed
\frac{2|r|^2}{|\det(B
  - M(\gl))|^2} \le 1, \quad \gl \in \dR \setminus
[-a,a].
\eed

Let $Q_\pm$ be the orthogonal projection from $L^2(\dR,\dC^2)$ onto
$L^2(\dR_\pm,\dC^2)$. Obviously, $Q_\pm$ commute with $H_0$. With
respect to the spectral representation the charges $Q_\pm$ correspond
to
\bed
Q_-(\gl) =
\begin{pmatrix}
1 & 0\\
0 & 0
\end{pmatrix}
\quad \mbox{and} \quad
Q_+(\gl) =
\begin{pmatrix}
0 & 0\\
0 & 1
\end{pmatrix},
\quad \gl \in \dR \setminus [-a,a].
\eed
If the steady state $\rho$ is chosen as
\bed
\rho = \rho_- \oplus \rho_+,
\eed
then the corresponding charge matrices are given by
\bed
\rho(\gl) =
\begin{pmatrix}
\rho_-(\gl) & 0 \\
0 & \rho_+(\gl)
\end{pmatrix},
\quad \gl \in \dR \setminus [-a,a].
\eed
where $\rho_\pm(\gl)$ are non-negative bounded Borel functions on $\dR
\setminus [-a,a]$. The operator $(I + H^2_0)\rho$ is bounded if and
only if $\ess-sup_{\gl\in \dR \setminus [-a,a]}(1 + \gl^2)\rho_\pm(\gl) <
\infty$. Applying Theorem \ref{thm:LandBuett} we find that the current
$J^\cS_{\rho, Q_-}(|r|)$ is given by
\begin{align*}
J^\cS_{\rho, Q_-}(|r|) & = \frac{1}{2\pi}\int_{\dR\setminus
  [-a,a]}(\rho_-(\gl) - \rho_+(\gl))\gs(\gl) d\gl\\
& = \frac{2|r|^2}{\pi}\int_{\dR\setminus
  [-a,a]}\frac{\rho_-(\gl) - \rho_+(\gl)}{|\det(B-M(\gl))|^2}d\gl
\end{align*}
A very simple case arises if we set $b_\pm = 0$. In this case we have
\bed
J^\cS_{\rho, Q_-}(|r|) = \frac{2|r|^2}{(1 + |r|^2)^2\pi}\int_{\dR\setminus
  [-a,a]}(\rho_-(\gl) - \rho_+(\gl))d\gl.
\eed
The magnitude of the current becomes maximal in this case if $|r| =1$,
that is, if 
\bed
J^\cS_{\rho, Q_-}(1) = \frac{1}{2\pi}\int_{\dR \setminus
  [-a,a]}(\rho_-(\gl) - \rho_+(\gl))d\gl.
\eed
Since $\gs(\gl) \le 2$ we find the estimate
\bed
|J^\cS_{\rho, Q_-}(|r|)| \le \frac{1}{\pi}\int_{\dR \setminus
  [-a,a]}(\rho_+(\gl) + \rho_-(\gl))d\gl.
\eed
Obviously $J^\cS_{\rho,Q_-}(0) = 0$ which is natural. 
In this case the self-adjoint operator $H$ decomposes into a left and
right hand side extension which have nothing to do with each other. 
However, one also
has $\lim_{|r|\to\infty}J^\cS_{\rho, Q_-}(|r|) = 0$.

For electrons one has to choose 
\bed
\rho_\pm(\gl) := \rho_{FD}(\gl - \mu_\pm), \quad \gl \in \dR,
\eed
where $\mu_\pm$ is the so-called Fermi energy and $\rho_{FD}(\gl)$
is the Fermi-Dirac distribution 
\bed
\rho_{FD}(\gl) = (1 + e^{\gb\gl})^{-1}, \quad \gl \in \dR, \quad \gb > 0. 
\eed
Obviously, the condition
$\ess-sup_{\dR\setminus[-a,a]}(1+\gl^2)\rho_{\pm}(\gl) < \infty$ is not satisfied. 
However, it turns out that
\bed
\rho_-(\gl) - \rho_+(\gl)  = e^{\gb\gl}(e^{-\gb\mu_+} -
e^{-\gb\mu_-})\rho_-(\gl)\rho_+(\gl), \quad \gl \in \dR.
\eed
satisfies $\ess-sup_{\dR \setminus [-a,a]}(1+\gl^2)|\rho_-(\gl) -
\rho_+(\gl)| < \infty$ which shows that the current
$J^\cS_{\rho,Q_-}$ is well defined.

\section*{Appendix: Spectral representations}

\setcounter{section}{0}
\renewcommand*\thesection{\Alph{section}}

\section{Spectral representation for unitary operators}\label{A}

Let $\sk$ be a separable Hilbert space and let $\mu$ a
Borel measure on the unit circle $\dT$. We consider the Hilbert space
$L^2(\dT,d\mu,\sk)$ and the multiplication operator $\cZ$ defined by
\bed
(\cZ \wh f)(\zeta) = \zeta \wh f(\zeta), \quad \wh f \in L^2(\dT,d\mu,\sk).
\eed 
Let $\{P(\zeta)\}_{\zeta\in\dT}$ be a measurable family of orthogonal projections in $\sk$. Setting
\be\label{A.168}
(P\wh f)(\zeta) = P(\zeta)\wh f(\zeta), \quad \wh f \in L^2(\dT,d\mu,\sk),
\ee
one defines orthogonal projection on $L^2(\dT,d\mu,\sk)$. The subspace $PL^2(\dT,d\mu,\sk)$ 
is denoted by $L^2(\dT,d\mu(\zeta),\sk(\zeta))$ where $\sk(\zeta) := P(\zeta)\sk$ 
in the following and is called a direct integral of Hilbert spaces $\{\sk(\zeta)\}_{\zeta\in\dT}$,
cf.\cite{BS1987}. We recall if an orthogonal projection on 
$L^2(\dT,d\mu,\sk)$ commutes with $\cZ$, then there is a measurable
family $\{P(\zeta)\}_{\zeta\in\dT}$ of orthogonal projections such
that $P$ is given by \eqref{A.168}.

For any unitary operator $U$  there is a
separable Hilbert space $\sk$ and a Borel measure $\mu$ on
$\dT$ such that $U$ is unitarily equivalent to a part of $\cZ$. That
means, there is an isometry $\Psi: \sH \longrightarrow
L^2(\dT,d\mu,\sk)$ such that
\bed
\Psi U = \cZ \Psi.
\eed
The operator $P = \Psi\Psi^*$ is an orthogonal projection on
$L^2(\dT,d\mu,\sk)$ commuting with $\cZ$. Hence there is a family of measurable 
orthogonal projections $\{P(\zeta)\}_{\zeta \in \dT}$ such that $P$ is given by \eqref{A.168}.
Notice that $\Psi$ is an isometry
acting from $\sH$ onto $L^2(\dT,d\mu,\sk)$. The multiplication operator $M :=
\cZ\upharpoonright L^2(\dT,d\mu(\zeta),\sk(\zeta))$, 
\bed
(M f)(\zeta) = \zeta f(\zeta), \quad f \in L^2(\dT,d\mu(\zeta),\sk(\zeta)),
\eed
is unitarily equivalent to $U$. The triplet $\Pi(U) = \{L^2(\dT,d\mu(\zeta),\sk(\zeta)),M,\Psi\}$
is called a spectral representation of $U$.

The existence of a spectral representation can be proved as
follows. Let $\mu(\cdot)$ be a scalar measure defined on $\sB(\dT)$
such that the spectral measure $E(\cdot)$ of $U$, 
\bed
U = \int_\dT \zeta dE(\zeta),
\eed
is absolutely continuous with respect to $\mu(\cdot)$. Such a measure
$\mu$ always exists. Indeed, let $C = C^*$ be a Hilbert-Schmidt operator such that
$\sH = \cH_C(U) := \clospa\{E(\gd)\ran(C): \gd \in \cB(\dT)\}$
where $E(\cdot)$ is the spectral measure of $U$. We set
\bed
\mu(\gd) := \tr(CE(\gd)C), \quad \gd \in \sB(\dT).
\eed
Obviously, the spectral measure $E(\cdot)$ is absolutely continuous
with respect to $\mu(\cdot)$. In fact, both measures are equivalent. 

Moreover, the operator-valued measure
$\gS(\gd) := CE(\gd)C$, $\gd \in \sB(\dT)$, is absolutely continuous
with respect to $\mu(\cdot)$ and takes values in $\sL_1(\sH)$. Since
$\sL_1(\sH)$ has the Radon-Nikodym property $\gS(\cdot)$ admits a Radon-Nikodym
derivative $\gY(\cdot)$ of $\gS(\cdot)$ exists with respect to $\mu(\cdot)$,
belongs to $\gY(\zeta) \in \sL_1(\sH)$ for a.e. $\zeta \in \dT$  and satisfies
$\gY(\zeta) \ge 0$ for a.e. $\zeta \in \dT$ with respect to $\mu$. Hence we have
\bed
\gS(\gd) =  \int_\gd \gY(\zeta)d\mu(\zeta)
\eed
for any Borel set $\gd \in \sB(\dT)$. We set $\sk(\zeta) :=
\ran(\gY(\zeta)) \subseteq \sk$, $\zeta \in \dT$, which defines a measurable family of
subspaces of $\sk := \ran(C)$. That means, the corresponding family of
orthogonal projections from $\sk$ onto $\sk(\zeta)$ is measurable with
respect to $\mu(\cdot)$. 
\bl	\label{A.I}
Let $L^2(\dT,d\mu(\zeta),\sk(\zeta))$ and $\gY(\zeta)$ be as above.
Further, let $\Psi$ be the linear extension of the mapping
\bed
\big(\Psi E(\gd) C f\big)(\zeta) = \chi_\gd(\zeta)\sqrt{\gY(\zeta)}f, \quad \zeta \in \dT, \quad f\in \sH.
\eed
If $\sH = \cH_C(U)$, 
then $\Pi(U) = \{L^2(\dT,d\mu(\zeta),\sk(\zeta)),M,\Psi\}$ is a spectral representation of $U$.
\el
\begin{proof}
Obviously, we have
\bed
\|\Psi E(\gd) C f\|^2_{L^2(\dT,d\mu(\zeta),\sk)} =
\int_\gd\|\sqrt{\gY(\zeta)}f\|^2_\sk d\mu(\zeta) = (\gS(\gd)f,f), \quad
f \in \sH.
\eed
Hence $\Psi$ is an isometry action from $\sH_C(U)$ into
$L^2(\dT,d\mu(\zeta),\sk)$ with range
$L^2(\dT,d\mu(\zeta),\sk(\zeta))$. Since $\sH = \sH_C(U)$ one gets an
isometry acting from $\sH$ onto
$L^2(\dT,d\mu(\zeta),\sk(\zeta))$. Moreover, by
\bed
(\Psi \int_\dT U dE(\zeta) C f)(\zeta)  = \zeta \sqrt{\gY(\zeta)}f,
\quad \zeta \in \dT,  \quad f\in \sH,
\eed
we get $\Psi U = \cZ\Psi$.
\end{proof}

The integer function $N_U : \dT
\longrightarrow \overline{\dN}_0 := \{0,1,2,\ldots,\infty\}$, $N_U(\zeta) :=
\dim(\sk(\zeta))$, is called the spectral
multiplicity function of $U$. We note that the family
$\{\sk(\zeta)\}_{\zeta \in \dT}$ and the spectral multiplicity
function $N_U$ are defined only a.e. with respect to $\mu$.
Furthermore,  it can happen that $\sk(\zeta) = \{0\}$ for $\zeta
\in \dT$ which yields $N_U(\zeta) = 0$. We set $\supp(N_U) :=
\{\zeta \in \dT: N_U(\zeta) > 0\}$ and introduce 
the measure $\mu_U := \chi_{\supp(N_U)}\mu$ which is absolutely
continuous with respect to $\mu$. 

Let $U$ and $\wt U$ be unitary
operators and let  $\Pi(U) =
\{L^2(\dT,d\mu,\sk(\zeta)),M,\Psi\}$ and 
$\wt \Pi(\wt U) = \{L^2(\dT,d\wt \mu(\zeta),\wt \sk(\zeta)),\wt M,\wt
\Psi\}$ be spectral representations, respectively. The operators $\wt U$ and $U$ are unitary
equivalent if and only if $\wt \mu_{\wt U}$ and $\mu_U$
are equivalent and $N_{\wt U}(\zeta) = N_U(\zeta)$ a.e. with respect
to $\mu_U$. The unitary operator $U$ is called of
constant spectral multiplicity $k \in \overline{\dN} := \{1,2,\ldots,\infty\}$ if $N_U(\zeta) = k$ a.e. with
respect to $\mu_U$.

\section{Spectral representation for $U^{ac}$}\label{B}

In the paper we mainly need a spectral representation of the
absolutely continuous part $U^{ac}$ of a unitary operator $U$. In this
case we choose $\mu =\nu$ where $\nu$ is the Haar measure on $\dT$. 
In this case the construction above simplifies as follows:

As above, let $C =C^* \in \sL_2(\sH)$ be a Hilbert-Schmidt operator on $\sH$.
Since $C \in \sL_2(\sH)$ we define by $\gS^{ac} := C
E_0^\textnormal{ac}(\cdot) C$ a $\sL_1$--valued measure on $\dT$
which is absolutely continuous with respect to the Haar measure $\nu$
on $\dT$. Its Radon-Nikodym derivative is denoted by $Y(\cdot)$.

Let us define a measurable family of subspaces 
by $\sh(\zeta)$ by setting $\sh(\zeta) := \clo \big\{\ran\big(Y(\zeta)\big)\big\} \subseteq
\sh$ in $\sh = \clo(\ran(C))$. With this family we associate the direct integral
$L^2(\dT,d\nu(\zeta),\sh(\zeta))$.
\bl\label{lem:3.0}
Let  $L^2(\dT,d\nu(\zeta),\sh(\zeta))$ and $Y(\zeta)$ as above. 
Further let $\Phi$ be the linear extension of the mapping
\bed
\big(\Phi E^\textnormal{ac}(\zeta) C f\big)(\zeta) = \chi_\gd(\zeta)\sqrt{Y(\zeta)}f, \quad \zeta \in \dT, \quad f\in \sH.
\eed
If the condition $\sH^{ac}(U) = \cH^{ac}_C :=
\clospa\{E^{ac}(\gd)\ran(C): \gd \in \cB(\dT) \}$ is satisfied, then 
$\Pi(U^{ac}) := \{L^2(\dT,d\nu(\zeta),\sh(\zeta)),M,\Phi\}$ defines a spectral representation of $U^{ac}$.
\el
The proof is similar to that one of Lemma \ref{A.I}. 
If the condition $\sH^{ac}_C =\cH^{ac}(U)$ is not satisfied, 
then $\Pi(U^{ac}) = \{L^2(\dT,d\nu(\zeta),\sh(\zeta)),M,\Phi\}$
is not a spectral representation of $U^{ac}$ but of 
$U^{ac}_C := U\upharpoonright\cH^{ac}_C$. Notice that $\sH^{ac}_C \subseteq \sH^{ac}(U)$
reduces $U^{ac}$. 

The following Lemma describes the action of the transformation 
$\Phi$ and is also valid for this extension of the spectral representation of Lemma \ref{lem:3.0}.
\bl\label{lem:3.1}
Let $X : \dT \rightarrow \sB(\sH)$ be strongly continuous. If the operator spectral integral
\bed	
L f = \int_{\dT}\de E^\textnormal{ac}(\zeta) C X(\zeta) f, \qquad f \in \sH,
\eed
exists, then
\be\label{eq:3.1}
(\Phi L f)(\zeta) = \sqrt{Y(\zeta)}X(\zeta)f, \qquad \zeta \in \dT, \qquad f \in \sH,
\ee
holds. Furthermore,
\bed
L^\ast f := \int_{\dT}X^\ast(\zeta)C \de E_0^\textnormal{ac}(\zeta)f
\eed
and
\be\label{eq:3.2}
L^\ast \Phi^\ast \widehat{f} = \int_\dT \de \lambda X^\ast(\zeta)\sqrt{Y(\zeta)}\widehat{f}(\zeta), 
\qquad f = \Phi^\ast \widehat{f} \in \sH^\textnormal{ac}.
\ee
\el
\begin{proof}
Let $\cJ_\epsilon$, $\epsilon > 0$, be a family of partitions of 
$\dT$ such that $\sup\limits_{\Xi \in \cJ_\epsilon} |\Xi| =
\epsilon$. Let further $\zeta_\epsilon : \cJ_\epsilon \rightarrow \dT$ 
satisfy $\zeta_\epsilon(\Xi) \in \Xi$ for all $\Xi \in \cJ_\epsilon$. Then for
\bed
Lf := \int_\dT\de E^\textnormal{ac}(\zeta) C X(\zeta) f, \qquad f \in \sH,
\eed
we have
\bed
Lf = \lim\limits_{\epsilon \rightarrow 0}
\sum\limits_{\Xi\in\cJ_\epsilon} E^\textnormal{ac}(\Xi)C X(\zeta_\epsilon(\Xi))f.
\eed
by definition. Since $\Phi_0$ is continuous and $\ran(L) \subset \cH(C)$, we have
\begin{align}
(\Phi Lf)(\lambda) &= \lim\limits_{\epsilon \rightarrow 0}
        	\sum\limits_{\Xi\in\cJ_\epsilon} \big(\Phi
                E^\textnormal{ac}(\Xi) C
                X(\lambda_\epsilon(\Xi))f\big)(\lambda)
\nonumber\\
&= \lim\limits_{\epsilon \rightarrow 0}	\sum\limits_{\Xi\in\cJ_\epsilon}\chi_{\Xi}(\lambda)\sqrt{Y(\lambda)} X(\lambda_\epsilon(\Xi))f
\nonumber
\end{align}
for a.e. $\zeta \in \dT$. Now let $\Xi_\epsilon(\lambda)$ be the
unique element in $\cJ_\epsilon$ for which $\lambda \in \Xi_\epsilon(\lambda)$. Since $X$ is continuous, we obtain
\bed
(\Phi L f)(\lambda) = \lim\limits_{\epsilon \rightarrow 0}\sqrt{Y(\lambda)}
        	X(\lambda_\epsilon(\Xi_\epsilon(\lambda)))f = \sqrt{Y(\lambda)} X(\lambda)f.
\eed
The adjoint relation \eqref{eq:3.2} follows easily from
\bead
\lefteqn{
\Big\langle g, \int_\dT X(\zeta) C \de E^\textnormal{ac}(\zeta)f \Big\rangle 
        	= \Big\langle \int_\dT \de E^\textnormal{ac}(\zeta) C
                X^\ast(\zeta)g, f \Big\rangle} = \\
& &
\int_\dT \de \lambda \big\langle \sqrt{Y(\zeta)}
        X^\ast(\zeta)g, (\Phi f)(\zeta) \big\rangle
= \Big\langle  g,  \int_\dT \de \zeta X(\zeta)\sqrt{Y(\zeta)} (\Phi f)(\zeta) \Big\rangle
\eead
for  all $g \in \sH$.
\end{proof}

\section{Spectral representation for $H^{ac}$}\label{C}

Let $H$ be a self-adjoint operator on the separable Hilbert space $\sH$. We introduce its Cayley transform
\bed
U := (i - H)(i + H)^{-1}.
\eed
Obviously, we have
\bed
E_U(\gd) = E_H(\gd'), \quad \gd \in \cB(\dT),\quad \gd' = \{\gl \in \dR: e^{2i\arctan(\gl)} \in \gd\}.
\eed
Let $\Pi(U^{ac}) = \{L^2(\dT,d\nu(\zeta),\sh(\zeta)),M,\Phi\}$.
Let us introduce the direct integral $L^2(\dR,d\gl,\sh'(\gl))$ where $d\gl$ is the Lebesgue measure on $\dR$,
and $\sh'(\gl) := \sh(e^{2i\arctan(\gl)})$. A straightforward
computation shows that the linear map
$F : L^2(\dT,d\nu(\zeta),\sh(\zeta)) \longrightarrow L^2(\dR,d\gl,\sh'(\gl))$,
\bed
\wh{f'}(\gl) := (F\wh f)(\gl) := \sqrt{\frac{2}{1+\gl^2}}\wh f(e^{2i\arctan(\gl)}), \quad \gl \in \dR,
\eed
$\wh f \in L^2(\dR,d\gl,\sh'(\gl))$, defines an isometry acting from $L^2(\dT,d\nu(\zeta),\sh(\zeta))$ onto $L^2(\dR,d\gl,\sh'(\gl))$.
Let $\{Q(\zeta)\}_{\zeta\in\dT}$ be a measurable operator-valued function which defines a multiplication operator $M_Q$ in the direct integral
$L^2(\dT,d\nu(\zeta),\sh(\zeta))$. Setting
\bed
Q'(\gl) = Q(e^{2i\arctan(\gl)}), \quad \gl \in \dR,
\eed
one easily defines a multiplication operator in $M_{Q'}$ in $L^2(\dR,d\gl,\sh'(\gl))$. It turns out that
$M_{Q'} = FM_QF^{-1}$. In particular, one gets that
\bed
F M_{\chi_\gd}F^{-1} = M_{\chi_\gd'}, \quad \gd \in \cB(\dT),\quad \gd' = \{\gl \in \dR: e^{2i\arctan(\gl)} \in \gd\}.
\eed
the last relation immediately shows that $\Pi(H^{ac}) := \{L^2(\dR,d\gl,\sh'(\gl)),M,\Phi'\}$, $\Phi' := F\Phi$,
defines a spectral representation of the absolutely continuous part
$H^{ac}$ of $H$.

\section{Scattering matrix for unitary operators}\la{D}

Let $\sH$ be a separable Hilbert space and let $U$ and $U_0$ be
unitary operators such that
\be\label{2.1}
V := U - U_0 \in \sL_1(\sH).
\ee
where $\sL_1(\cdot)$ denotes the set of trace class operators in
$\sH$. In the following we call the pair $\cS =\{U,U_0\}$ of unitary
operators satisfying \eqref{2.1} a $\sL_1$-scattering system.

If $\cS = \{U,U_0\}$ is a $\sL_1$ scattering system, then the wave operators
\bed
\gO_\pm := \gO_\pm(U,U_0) := \slim_{n \to\pm\infty}U^nU^{-n}_0P^{ac}(U_0)
\eed
exist and are complete. Completeness means that $\ran(\gO_\pm) =
\sH^{ac}(U)$ where $\sH^{ac}(U)$.
The scattering operator $S$ of the scattering system $\cS$ is defined by
\bed
S := S(U,U_0) := \gO^*_+\gO_-.
\eed
In fact, the scattering operator acts only on $\sH^{ac}(U_0)$ and is
unitary there. Moreover, it commutes with $U_0$. 

Let $\Pi(U^{ac}_0) = \{L^2(\dT,d\nu(\zeta),\sh(\zeta)),M,\Phi)$ be a spectral
representation  of the absolutely continuous part $U^{ac}_0$ of $U_0$,
cf. Appendix \ref{B}.
Since the scattering operator $S$ is unitary on $\sH^{ac}(U_0)$ and
commutes with $U^{ac}_0$ there is a measurable family
$\{S(\zeta)\}_{\zeta \in \dT}$ of unitary operator on
$\sh(\zeta)$ such that $S$ is unitary equivalent to $M_S$,
\bed
(M_Sf)(\zeta) = S(\zeta)f(\zeta), \quad f \in
L^2(\dT,d\nu(\zeta),\sh(\zeta)),
\eed
that is $S = \Phi^{-1} M_S \Phi$. 
The family $S(\zeta)$ of unitary operators is called the scattering matrix of the scattering system $\cS$.

At first we prove a technical lemma.
\bl\label{II.6}
Let $\cS = \{U,U_0\}$ be $\sL_1$-scattering system. Then there is a bounded
self-adjoint Hilbert-Schmidt operator $C$ and a bounded operator $G$
such that the representation
\be\label{2.2}
V := U - U_0 = CGC
\ee
is valid.
\el
\begin{proof}
Let $V = V_R + iV_I$ where where $V_R := \frac{1}{2}(V + V^*)$ and
$V_I := \frac{1}{2i}(V^* - V^*)$. Obviously, one has $V_R := V^*_R \in
\sL_1(\sH)$. and $V_I = V^*_I \in \sL_1(\sH)$. Let $C_R :=
|V_R|^{1/2}$ and $C_I := |V_I|^{1/2}$. Then
\be\label{2.3}
V_R = C_R G_R C_R \quad \mbox{and} \quad
V_I = C_I G_I C_I
\ee
where $G_R := \sign(V_R)$ and $G_I := \sign(V_I)$. We set
\bed
C := (|V_R| + |V_I|)^{1/2}.
\eed
Obviously, we have 
\bed
\|C_R f\|^2 = (|V_R|f,f) \le ((|V_R| + |V_I|)f,f) = \|C f||^2, \quad f
\in \sH.
\eed
Hence there is  a contraction $\gG_R$ such that $C_R = \gG_R C$ and $C_R
= C\gG^*_R$. Similarly, there is a contraction $\gG_I$ such that 
$C_I = \gG_I C_I$ and $C_I = C_I \gG^*_I$. From \eqref{2.3} we find
\bed
V = C(\gG^*_RG_R\gG_R + i\gG^*_I  G_I \gG_I)C.
\eed
Setting $G := \gG^*_RG_R\gG_R + i\gG^*_I  G_I \gG_I$ we prove
\eqref{2.2}.
\end{proof}

We define the Abel pre-wave operators by 
\begin{equation} \label{mitlef23}
\begin{matrix}
\gO_+(r) & := & (1-r)\sum^\infty_{n=0}r^nU^{n}U^{-n}_0P^{ac}_0,\\[2mm]
\gO_-(r) & := & (1-r)\sum^\infty_{n=0}r^nU^{-n}U^n_0P^{ac}_0,
\end{matrix}
\end{equation}
$r \in [0,1)$, where we have used the abbreviation $P^{ac}_0 :=
P^{ac}(U_0)$. It holds 
\bed
\gO_\pm = \slim_{r\uparrow 1}\gO_\pm(r).
\eed
Let $E_0(\cdot)$ be spectral measure of $U_0$ defined on the Borel
subsets of $\dT$. We set $E^{ac}_0(\cdot) := P^{ac}(U_0)E_0(\cdot)$.
A straightforward computation gives
\bea
\gO_+(r) & := & P^{ac}_0 + r\int_\dT \frac{\overline{\zeta}}{I - r\overline{\zeta}U}VdE^{ac}_0(\zeta),\label{2.57}\\
\gO_-(r) & := & P^{ac}_0 - r\int_\dT\frac{U^*}{I - r\zeta U^*}VdE^{ac}_0(\zeta).\label{2.53}
\eea
Using $U^*V = -V^*U_0$ we find
\be\la{2.54}
\gO_-(r) := P^{ac}_0 + r\int_\dT\frac{\zeta}{I - r\zeta U^*}V^*dE^{ac}_0(\zeta).
\ee
Furthermore, from \eqref{2.57} and \eqref{2.54} we get
\bea
\gO_+(r)^* & = & P^{ac}_0 + r\int_\dT dE^{ac}_0(\zeta)V^*\frac{\zeta}{I - r\zeta U^*}\label{2.54a}\\
\gO_-(r)^* & = & P^{ac}_0 + r\int_\dT dE^{ac}_0(\zeta)V\frac{\overline{\zeta}}{I - r\overline{\zeta} U}\label{2.60a}
\eea
Notice that $\wlim_{r\uparrow 1}\gO_+(r)^* = \gO^*_+$. Similarly, we
find the representations
\bead
\gO_+(r) & = & P^{ac}_0 + r\int_\dT dE(\zeta)V\frac{U^*_0}{I-r\zeta U^*_0}P^{ac}_0\\
\gO_-(r) & = & P^{ac}_0 - r\int_\dT dE(\zeta)V\frac{\overline{\zeta}}{I - r\overline{\zeta}U_0}P^{ac}_0.
\eead
Using again $U^*V = -V^*U_0$ we get
\bea
\gO_+(r) & = & P^{ac}_0 - r\int_\dT dE(\zeta)V^*\frac{\zeta}{(I-r\zeta U^*_0)}P^{ac}_0\label{2.61a}\\
\gO_-(r) & = & P^{ac}_0 + r\int_\dT dE(\zeta)V^*\frac{U_0}{I - r\overline{\zeta}U_0}P^{ac}_0.\la{2.63a}
\eea
We consider the transition operator $T := \frac{1}{2i\pi}(P^{ac}_0 - S)$. Notice that
\bed
S = P^{ac}(U_0) - 2\pi i T.
\eed
In fact the operator $T$ acts only on $\sH^{ac}(U_0)$. 
Since the scattering operator $S$ commutes with $U_0$ the transition operator $T$ also commutes with $U_0$.
With respect to the spectral representation 
$\Pi(U^{ac}_0) = \{L^2(\dT,d\nu(\zeta),\sh(\zeta)),M,\Phi\}$
the transition operator $T$ takes the form of a multiplication
operator $M_T$ induced by a measurable family 
$\{T(\zeta)\}_{\zeta \in \dT}$ of bounded operators. Obviously, we have
\be\label{2.62a}
S(\gl) = I_{\sh(\zeta)} - 2\pi i T(\zeta)
\ee
for a.e. $\zeta \in \dT$. The family $T(\zeta)$ of bounded operators is called the transition matrix
of the scattering system $\cS$ . 
We are going to compute the measurable family $\{T(\zeta)\}_{\zeta\in \dT}$.
\bt\label{II.7}
Let $\cS = \{U,U_0\}$ be a $\sL_1$-scattering system.
With respect to the spectral representation $\Pi(U^{ac}_0) =
\{L^2(\dT,d\nu(\zeta),\sh(\zeta)),M,\Phi\}$ of $U^{ac}_0$,
cf. Appendix \ref{B},
the family of transition matrices $\{T(\zeta)\}_{\zeta\in\dT}$ admits the representation
\be\label{2.61}
T(\zeta) = i\zeta\sqrt{Y(\zeta)}Z(\zeta)\sqrt{Y(\zeta)}
\ee
for a.e. $\zeta \in \dT$ with respect to $\nu$ where $Z(\zeta) :=
\olim_{r\uparrow 1}Z(r\zeta)$ and 
\be\label{2.62}
Z(\xi) := G^* + G^*C\frac{\xi}{I - \xi U^*}CG^*, \quad \xi \in \dD := \{z\in \dC: |\xi| < 1\}.
\ee 
\et
\begin{proof}
Obviously we have
\bed 
T = \frac{1}{2i\pi}\gO^*_+(\gO_+ - \gO_-).
\eed
We set
\bed
T(r) = \frac{1}{2i\pi}\gO^*_+(\gO_+(r) - \gO_-(r)).
\eed
Notice that $T = \slim_{r\uparrow 1}T(r)$. Using the representations \eqref{2.61a} and \eqref{2.63a}
we get
\bed
T(r) = i\frac{r}{2\pi}\gO^*_+
\left\{\int_\dT dE(\xi)V^*\frac{\xi}{I-r\xi U^*_0} + 
\int_\dT dE(\xi)V^*\frac{U_0}{I - r\overline{\xi}U_0}\right\}P^{ac}_0
\eed
which yields
\bed
T(r) = i\frac{r}{2\pi(1+r)}
(1-r^2)\int_\dT dE^{ac}_0(\xi)\gO^*_+V^*\frac{U_0 + \xi}{|I-r\xi U^*_0|^2}P^{ac}_0.
\eed
Let us introduce the notation
\be\label{2.67}
T(r,s) := i\frac{r}{1+r}\frac{1-r^2}{2\pi}
\int_\dT dE^{ac}_0(\xi)\gO^*_+(s)V^*\frac{U_0 + \xi}{|I-r\xi U^*_0|^2}P^{ac}_0, \quad 0 \le r,s < 1.
\ee
Since $\wlim_{s\uparrow 1}\gO^*_+(s) = \gO^*_+$ it seems natural to expect that
$\wlim_{s\uparrow 1}T(r,s) = T(r)$ for $0 \le r < 1$. Indeed, integrating by parts we get
\bead
\lefteqn{
\int_\dT dE^{ac}_0(\xi)\gO^*_+ V^*\frac{U_0 + \xi}{|I-r\xi U^*_0|^2}P^{ac}_0 =}\\
& &
\gO^*_+V^*\frac{U_0 -1}{|I+rU^*_0|^2}P^{ac}_0 - 
\int_\dT E^{ac}_0(\xi)\gO^*_+V^*\frac{\partial}{\partial \xi}\frac{U_0 + \xi}{|I-r\xi U^*_0|^2}P^{ac}_0d\nu(\xi)\nonumber
\eead
and
\bead
\lefteqn{
\int_\dT dE^{ac}_0(\xi)\gO^*_+(s) V^*\frac{U_0 + \xi}{|I-r\xi U^*_0|^2}P^{ac}_0 =}\\
& &
\gO^*_+(s)V^*\frac{U_0 -1}{|I+rU^*_0|^2}P^{ac}_0 - 
\int_\dT E^{ac}_0(\xi)\gO^*_+(s)V^*\frac{\partial}{\partial \xi}\frac{U_0 + \xi}{|I-r\xi U^*_0|^2}P^{ac}_0d\nu(\xi).\nonumber
\eead
Because $\frac{\partial}{\partial \xi}\frac{U_0 + \xi}{|I-r\xi
  U^*_0|^2}$ is bounded for $r \in[0,1)$ we find that
\bead
\lefteqn{
\wlim_{s\uparrow 1}\int_\dT dE^{ac}_0(\xi)\gO^*_+(s)V^*\frac{U_0 + \xi}{|I-r\xi U^*_0|^2}P^{ac}_0 =}\\
& &
\gO^*_+V^*\frac{U_0 -1}{|I+rU^*_0|^2}P^{ac}_0 - 
\int_\dT E^{ac}_0(\xi)\gO^*_+V^*\frac{\partial}{\partial \xi}\frac{U_0 + \xi}{|I-r\xi U^*_0|^2}P^{ac}_0d\nu(\xi)
\nonumber
\eead
which proves $\wlim_{s\uparrow 1}T(r,s) = T(r)$ for $0 \le r < 1$.
From \eqref{2.54a} we get
\be\label{2.71}
\gO_+(s)^* = \int_\dT dE^{ac}_0(\zeta)\left\{I + s\,V^*\frac{\zeta}{I - s\zeta U^*}\right\}.
\ee
Inserting \eqref{2.71} into \eqref{2.67} we obtain
\bed
T(r,s) = i\frac{r}{1+r}
\frac{1-r^2}{2\pi}
\int_\dT dE^{ac}_0(\zeta)\left\{I + s\,V^*\frac{\zeta}{I - s\,\zeta U^*}\right\}
V^*\frac{U_0 + \zeta}{|I-r\zeta U^*_0|^2}P^{ac}_0
\eed
where $\zeta \in \dT$. Using \eqref{2.2} and the notation \eqref{2.62} we get
\be\label{2.72}
T(r,s) = 
i\frac{r}{1+r}\frac{1-r^2}{2\pi}
\int_\dT dE^{ac}_0(\zeta)CZ(s\zeta) C\frac{U_0 + \zeta}{|I-r\zeta U^*_0|^2}P^{ac}_0.
\ee
Inserting the representation 
\bed
\frac{U_0 + \zeta}{|I-r\zeta U^*_0|^2}P^{ac}_0 = 
\int_\dT\frac{\xi + \zeta}{|1- r\zeta\overline{\xi}|^2}dE^{ac}_0(\xi)
\eed
into \eqref{2.72} we find
\bed
T(r,s) = 
i\frac{r}{1+r}
\int_\dT dE^{ac}_0(\zeta)CZ(s\zeta)
\frac{1-r^2}{2\pi}\int_\dT \frac{\xi + \zeta}{|1- r\zeta\overline{\xi}|^2}C dE^{ac}_0(\xi)
\eed
which leads to
\bed
T(r,s) = 
i\frac{r}{1+r}
\int_\dT dE^{ac}_0(\zeta)CZ(s\zeta)\zeta\;\;
\frac{1-r^2}{2\pi}\int_\dT \frac{\xi\overline{\zeta} + 1}{|1- r\zeta\overline{\xi}|^2}CdE^{ac}_0(\xi).
\eed
Applying the map $\Phi : \sH^{ac}(U_0) \longrightarrow L^2(\dT,d\nu(\zeta),\sh(\zeta))$ we obtain 
\bead
\lefteqn{
(\Phi T(r,s)\Phi^{-1}\wh f)(\zeta) =}\\
& &
i\frac{r}{1+r}\sqrt{Y(\zeta)}Z(s\zeta)\zeta\;
\frac{1-r^2}{2\pi}\int_\dT\frac{\xi\overline{\zeta} + 1}{|1- r\zeta\overline{\xi}|^2}\sqrt{Y(\xi)}\wh f(\xi)d\nu(\xi)
\nonumber
\eead
where $\wh f \in L^2(\dT,d\nu(\zeta),\sh(\zeta))$. We set 
\be\la{2.41}
X(s;\zeta) := \sqrt{Y(\zeta)}Z(s\zeta), \quad \zeta \in \dT, \quad 0
\le s < 1. 
\ee
Notice that $X(s;\zeta) \in \sL_2(\sH)$ for a.e. $\zeta \in \dT$. Since
$X(s) := \sL_2-\lim_{s\uparrow 1}X(s;\zeta) = \sqrt{Y(s)}Z(s)$ exists for
a.e. $\zeta \in \dT$ there is a Borel subset $\gD(\varepsilon) \subseteq \dT$ 
for every $\varepsilon > 0$ such that
$\nu(\gD(\varepsilon)) < \varepsilon$ and 
\be\la{2.41a}
C_X(\varepsilon) := \sup\left\{\|X(s;\zeta)\|_{\sL_2}: \zeta\in \dT \setminus \gD(\varepsilon), \quad 
0 \le s < 1\right\}  < \infty
\ee
is valid. We note the existence of the set $\gD(\varepsilon)$
follows from Egorov's theorem. 

Using that observation we get
\bead
\lefteqn{
(\Phi E^{ac}_0(\dT \setminus \gD(\varepsilon))T(r) \Phi^{-1}\wh f)(\zeta) =
 \wlim_{s\uparrow 1}(\Phi
E^{ac}_0(\dT \setminus \gD)(\varepsilon))T(r,s)\Phi^{-1}\wh f)(\zeta)} \nonumber\\
& &
= i\zeta\frac{r}{1+r}\chi_{\dT\setminus\gD(\varepsilon)}(\zeta)\sqrt{Y(\zeta)}Z(\zeta)\;
\frac{1-r^2}{2\pi}\int_\dT\frac{\xi\overline{\zeta} + 1}{|1- r\zeta\overline{\xi}|^2}\sqrt{Y(\xi)}\wh f(\xi)d\nu(\xi)
\nonumber 
\eead
for a.e. $\zeta \in \dT$ with respect to $\nu$ and $\wh f \in L^2(\dT,d\nu(\zeta),\sh(\zeta))$.  Finally, taking the limit $r \uparrow 1$ we get
\bead
\lefteqn{
(\Phi  E^{ac}_0(\dT \setminus \gD(\varepsilon)) T \Phi^{-1}\wh f)(\zeta) = }\\
& &
\lim_{r\uparrow 1}\Phi  E^{ac}_0(\dT \setminus \gD(\varepsilon))T(r)\Phi^{-1}\wh f)(\zeta) =
i\zeta\chi_{\dT\setminus\gD(\varepsilon)}(\zeta)\sqrt{Y(\zeta)}Z(\zeta)\;\sqrt{Y(\zeta)}\wh f(\zeta)
\nonumber
\eead
for a.e. $\zeta \in \dT$ with respect to $\nu$ and $\wh f \in
L^2(\dT,d\nu(\zeta),\sh(\zeta))$ where  it was used that 
\bed 
\wh g(\zeta) = \frac{1}{2}\lim_{r\uparrow
  1}\frac{1-r^2}{2\pi}\int_\dT\frac{\xi\overline{\zeta} + 1}{|1-
  r\zeta\overline{\xi}|^2}\wh g(\xi)d\nu(\xi),\quad
\wh g \in L^2(\dT,d\nu(\zeta),\sh(\zeta)),
\eed
in the $L^2$-sense, see \cite[Section I.D.2]{Koosis1980}. 
If $\wh f(\zeta) \in L^\infty(\dT,d\nu(\zeta),\sh(\zeta))$, then
$\sqrt{Y(\zeta)}\wh f(\zeta) \in
L^2(\dT,d\nu(\zeta),\sh(\zeta))$.
Hence we find that
\bed
(T(\zeta)\wh f)(\zeta) =
i\zeta\sqrt{Y(\zeta)}Z(\zeta)\;\sqrt{Y(\zeta)}\wh
f(\zeta), 
\eed
for a.e. $\zeta \in \dT \setminus \gD(\varepsilon)$ and $f \in
L^\infty(\dT,d\nu(\zeta),\sh(\zeta))$ which yields \eqref{2.61} for a.e. $\zeta
\in\dT \setminus \gD(\varepsilon)$. Since $\varepsilon$ can be chosen
arbitrary small we we prove \eqref{2.61}. 
\end{proof}

From \eqref{2.62a} and \eqref{2.61} we get that the scattering matrix
admits the representation
\bed
S(\zeta) =  I_{\sh(\zeta)} + 2\pi \zeta
\sqrt{Y(\zeta)}Z(\zeta)\sqrt{Y(\zeta)}
\eed
for a.e. $\zeta \in \dT$. Since $\|S(\zeta)\|_{\sh(\zeta)} = 1$ for
a.e. $\zeta \in \dT$ we get $\|S(\zeta) -
I_{\sh(\zeta)}\|_{\sh(\zeta)} \le 2$ for a.e $\zeta \in \dT$ which
yields
\bed
\|\sqrt{Y(\zeta)}Z(\zeta)\sqrt{Y(\zeta)}\|_{\sh(\zeta)} \le
\frac{1}{\pi}
\eed
for a.e. $\zeta \in \dT$. In fact, this estimate can be proved directly.
\bc\label{II.8}
Let the assumptions of Theorem \ref{II.7} be satisfied. Then the
following holds:

\begin{enumerate}
\item[\em (i)] 
For $\wh f \in L^2(\dT,d\nu(\zeta),\sh(\zeta))$ we have 
\be\label{2.87}
(\Phi \gO^*_-(r)VP^{ac}(U_0)\Phi^{-1}\wh f)(\zeta) = \int_\dT
K(r;\zeta,\xi) \wh f(\xi)d\nu(\xi),
\quad r \in [0,1), 
\ee
for a.e. $\zeta \in \dT$ with respect to $\nu$ 
where $K(r;\zeta,\xi) := \sqrt{Y(\zeta)}Z(r\zeta)^*\sqrt{Y(\xi)}$, $\zeta,\xi \in \dT$.

\item[\em (ii)] 
For $\wh f \in L^2(\dT,d\nu(\zeta),\sh(\zeta))$ we have 
\be\label{2.81}
(\Phi\gO^*_-VP^{ac}(U_0)\Phi^{-1} \wh f)(\zeta) = \int_\dT K(\zeta,\xi)\wh f(\xi)d\nu(\xi).
\ee
for a.e. $\zeta \in \dT$ with respect to $\nu$ 
where $K(\zeta,\xi) := \sqrt{Y(\zeta)}\,Z(\zeta)^*\sqrt{Y(\xi)}$, $\zeta,\xi \in \dT$.

\item[\em (iii)] 
For a.e. $\zeta \in \dT$ with respect to $\nu$ 
one has the representation $T(\zeta) = i\zeta\,
K(\zeta,\zeta)^*$. Moreover, $T(\zeta) \in \sL_1(\sh(\gl))$ for a.e
$\zeta \in \dT$  with respect to $\nu$, $\|T(\zeta)\|_{\sS_1} \in L^1(\dT,d\nu(\zeta))$
and
\be\la{2.81a}
\int_\dT\|T(\zeta)\|_{\sL_1}d\nu(\zeta) \le \|V\|_{\sL_1}.
\ee
In addition one has
\be\label{2.82}
\tr(\gO^*_-V) = \int_\dT \tr(K(\zeta,\zeta))d\nu(\zeta) = i\int_\dT \zeta\,\tr(T(\zeta)^*)d\nu(\zeta).
\ee
\end{enumerate}
\ec
\begin{proof}
(i) Let $K(r) := \gO^*_-(r)V$. Using \eqref{2.60a} we get
\bed
K(r)P^{ac}_0 = \left\{P^{ac}_0 + r\int_\dT dE^{ac}_0(\zeta)V\frac{\overline{\zeta}}{I - r\overline{\zeta} U}\right\}VP^{ac}_0
\eed
which leads to
\bed
K(r)P^{ac}_0 = \int_\dT dE^{ac}_0(\zeta)C \left\{G + r GC\frac{\overline{\zeta}}{I - r\overline{\zeta} U}CG\right\}\int_\dT CdE^{ac}_0(\xi).
\eed
From \eqref{2.62} we get 
\bed
Z(r\zeta)^* = G + r GC\frac{\overline{\zeta}}{I - r\overline{\zeta} U}CG
\eed
which yields
\bed
K(r)P^{ac}_0 = \int_\dT dE^{ac}_0(\zeta)C Z(r\zeta)^*\int_\dT CdE^{ac}_0(\xi).
\eed
Thus
\bed
(\Phi  K(r)P^{ac}_0 \Phi^{-1} \wh f)(\zeta) = \sqrt{Y(\zeta)}Z(r\zeta)^*\int_\dT \sqrt{Y(\xi)}\wh f(\xi)d\nu(\xi),
\eed
$\wh f \in L^2(\dT,d\nu(\zeta),\sh(\zeta))$ which verifies \eqref{2.87}. 

(ii) Following the proof of Theorem \ref{II.7} we set
\be\la{2.55}
X_*(r;\zeta) := \sqrt{Y(\zeta)}Z(r\zeta)^*, \quad \zeta \in
\dT, \quad 0 \le r < 1. 
\ee
As above, using the existence of $X_*(\zeta) := \sL_2-\lim_{r\uparrow
  1}X(r;\zeta) = \sqrt{Y(\zeta)}Z(\zeta)^*$
for a.e. $\zeta \in \dT$ with respect to $\nu$ we find that for each $\varepsilon
> 0$ there is a Borel subset $\gD_*(\varepsilon) \subseteq \dT$ satisfying
$\nu(\gD_*(\varepsilon)) < \varepsilon$ such that the condition
\be\la{2.55a}
C_{X_*}(\varepsilon) := \sup\left\{\|X_*(s;\zeta)\|_{\sL_2}: 
\zeta\in \dT \setminus \gD_*(\varepsilon), \quad 0 \le s < 1\right\} < \infty.
\ee
Using $ K := \wlim_{r\uparrow 1}K(r) = \gO^*_-V$ we get
\bead
\lefteqn{
(\Phi E^{ac}(\dT\setminus\gD_*(\varepsilon))\gO^*_-VP^{ac}_0\Phi^{-1}\wh f)(\zeta) =}\\
& &
\wlim_{r\uparrow 1}(\Phi E^{ac}(\dT\setminus\gD_*(\varepsilon))\gO^*_-(r)
VP^{ac}_0\Phi^{-1}\wh f)(\zeta) = \nonumber\\
& &
\chi_{\dT\setminus\gD_*(\varepsilon)}(\zeta)\sqrt{Y(\zeta)}Z(\zeta)^*\int_\dT
\sqrt{Y(\xi)}\wh f(\xi)d\nu(\xi),
\nonumber
\eead
$\wh f \in L^2(\dT,d\nu(\zeta),\sh(\zeta))$, which proves \eqref{2.81}
for a.e. $\zeta \in \dT \setminus \gD_*(\varepsilon)$ with respect to
$\nu$. Since $\varepsilon$ is arbitrary \eqref{2.81} holds for a.e. $\zeta
\in \dT$. 

(iii) 
By \cite[Proposition 7.5.2]{Yafaev1992} we find that
$\|K(\zeta,\zeta)\|_{\sL_1} \in L^1(\dT,d\nu(\zeta))$ and
\bed
\int_\dT \|K(\zeta,\zeta)\|_{\sL_1}d\nu(\zeta) \le \|K\|_{\sL_1}.
\eed
From \eqref{2.61} we get that $T(\zeta) = i\zeta K(\zeta,\zeta)$
for a.e. $\zeta \in \dT$ with respect to $\nu$. Thus
$\|T(\zeta)\|_{\sH_1} \in L^1(\dT,d\nu(\zeta))$ and \eqref{2.81a} is valid.
Using again \cite[Proposition 7.5.2]{Yafaev1992} we find 
\bed
\tr(\gO^*_-V) = \tr(K)  = \int_\dT \tr(K(\zeta,\zeta))d\nu(\zeta).
\eed
By $T(\zeta) = i\zeta K(\zeta,\zeta)$ we prove \eqref{2.82}. 
\end{proof}

\section*{Acknowledgments}

\noindent
The first author acknowledges partial support from the Danish 
FNU grant {\it Mathematical Analysis of Many-Body Quantum Systems}. 
The second and third author thank the Aalborg University and the
Centre de Physique Th{\'e}orique in Marseille for hospitality and financial support. 


\end{document}